\documentclass[11pt]{article}
\usepackage{amssymb}
\usepackage{amsmath}
\usepackage{amsthm}
\usepackage{graphicx}
\usepackage{caption}
\usepackage{graphicx}
 \usepackage{subcaption}
\usepackage{multirow}
\usepackage{amsfonts}
\usepackage{floatrow}
 \usepackage{epsfig}
 \usepackage{graphicx}
\usepackage{xcolor}
\usepackage{cite}
\usepackage{breqn}
 \usepackage[utf8]{inputenc}
 \usepackage{float}
\usepackage{hyperref}
\usepackage{dsfont}
\def\bc{\begin{center}}
	\def\ec{\end{center}}

\def\s2c{\vskip 2cm}
\def\bt{\begin{Theorem}}
	\def\et{\end{Theorem}}
\def\bd{\begin{Definition}}
	\def\ed{\end{Definition}}
\def\bl{\begin{Lemma}}
	\def\el{\end{Lemma}}
\def\be{\begin{Example}}
	\def\ee{\end{Example}}
\def\bcor{\begin{Corollary}}
	\def\ecor{\end{Corollary}}
\def\br{\begin{Remark}}
	\def\er{\end{Remark}}

\def\mysection{\setcounter{equation}{0}\section}
\newtheorem{Lemma}{Lemma}[section]
\newtheorem{Theorem}[Lemma]{Theorem}
\newtheorem{Definition}[Lemma]{Definition}

\newtheorem{Corollary}[Lemma]{Corollary}
\newtheorem{Remark}[Lemma]{Remark}
\date{} 

\title{Analysis of a mathematical model for malaria using data-driven approach} 
\author{Adithya Rajnarayanan$^{1}$, Manoj Kumar$^{1}$, Abdessamad Tridane$^{2,*}$  \\
$^{1}$School of Engineering and Science \\
Indian Institute of Technology Madras, Zanzibar \\
PO Box 394, Bweleo,
Urban West - 71215 \\
Zanzibar, Tanzania \\ 
$^{2}$Department of Mathematical Sciences \\ United Arab Emirates University\\ P.O. Box 15551, Al Ain, UAE  \\
Authors' email IDs: manoj@iitmz.ac.in, zda23m017@iitmz.ac.in, a-tridane@uaeu.ac.ae \\ 
$*$Corresponding author's email: a-tridane@uaeu.ac.ae}

\begin{document}
\maketitle
	\author
	\noindent {\bf Abstract}: 
    Malaria is one of the deadliest diseases in the world, every year millions of people become victims of this disease and many even lose their lives. Medical professionals and the government could take accurate measures to protect the people only when the disease dynamics are understood clearly. In this work, we propose a compartmental model to study the dynamics of malaria. We consider the transmission rate dependent on temperature and altitude. We performed the steady state analysis on the proposed model and checked the stability of the disease-free and endemic steady state. Since the dynamics of a system are influenced by the parameters, we use three different architectures of neural network namely ANNs(artificial neural networks), RNNs(recurrent neural networks), and PINNs(physics-informed neural networks) to estimate the parameters of the SIR-SI dynamical system and then using the estimated parameters, trajectories of the compartments are predicted. To understand the severity of a disease, it is essential to calculate the risk associated with the disease. In this work, the risk is calculated using dynamic mode decomposition(DMD) from the trajectory of the infected people.   
	\footnote{$^*$Corresponding author's email}
%%%%%%%%%%%%%%%%%%%%%%%%%%%%%%%%%%
%%%%%%%%%%%%%%%%%%%%%%%%%%%%%%%%%%
\mysection{Introduction} 
According to the World Health Organization (WHO) report from 2022, Africa bears the highest burden of malaria among all regions. Notably, 94 \% of the global malaria cases, totaling 233 million, occur in Africa, accounting for 95 \% of the worldwide fatalities, which number approximately 58,000 (see Figure \ref{bar11}). These statistics highlight the significant mortality impact of malaria on the African continent, underscoring the urgent need to develop models that enable healthcare professionals to better understand the dynamics of this disease.

Differential equation-based compartmental models are commonly employed to analyze disease transmission. This approach divides the population into non-overlapping sequential compartments, allowing for the study of the movement of individuals between compartments using ordinary differential equations. Some of the most widely utilized compartments include SIR (Susceptible, Infected, Recovered), SIRD (Susceptible, Infected, Recovered, Death), and SIRDV (Susceptible, Infected, Recovered, Dead, Vaccinated).
% Out of all these the simplest and most used model is the SIR model. Analyzing a population using compartmental models has a lot of advantages, which include reducing the complexity of disease transmission, and providing a framework for predicting future states, these models are extremely flexible and can adapt extra compartments in the future if there is a requirement, compartmental models provide a clear visualization of the system's dynamics.

One of the earliest mathematical models for malaria transmission was developed by Ross, comprising two compartments: one for infected humans and another for infected mosquitoes. Key parameters incorporated in this model includes the biting rate, recovery rate, and death rate. Notably, Ross’s model did not account for the parasite's latent period, a critical factor in malaria transmission dynamics. For further details on Ross's work, readers can refer to \cite{ref20,ref21,ref22,ref24}.
Subsequently, MacDonald advanced the modeling framework by introducing a three-compartment model that included both susceptible and infected compartments for mosquitoes, alongside a single infected compartment for humans. Although this model considered the latent period of parasites in mosquitoes, it overlooked the latent period within humans, which represents a significant limitation. More information on MacDonald's model can be found in \cite{ref25}. Finally, the Anderson model expanded this approach to four compartments, incorporating susceptible and infected compartments for both species. For additional insights into this model, one can refer to \cite{ref26}. For a comprehensive overview of mathematical models of malaria, the review article in \cite{ref0} is recommended.

\begin{figure}[htbp]
    \centering
    \includegraphics[scale=0.4]{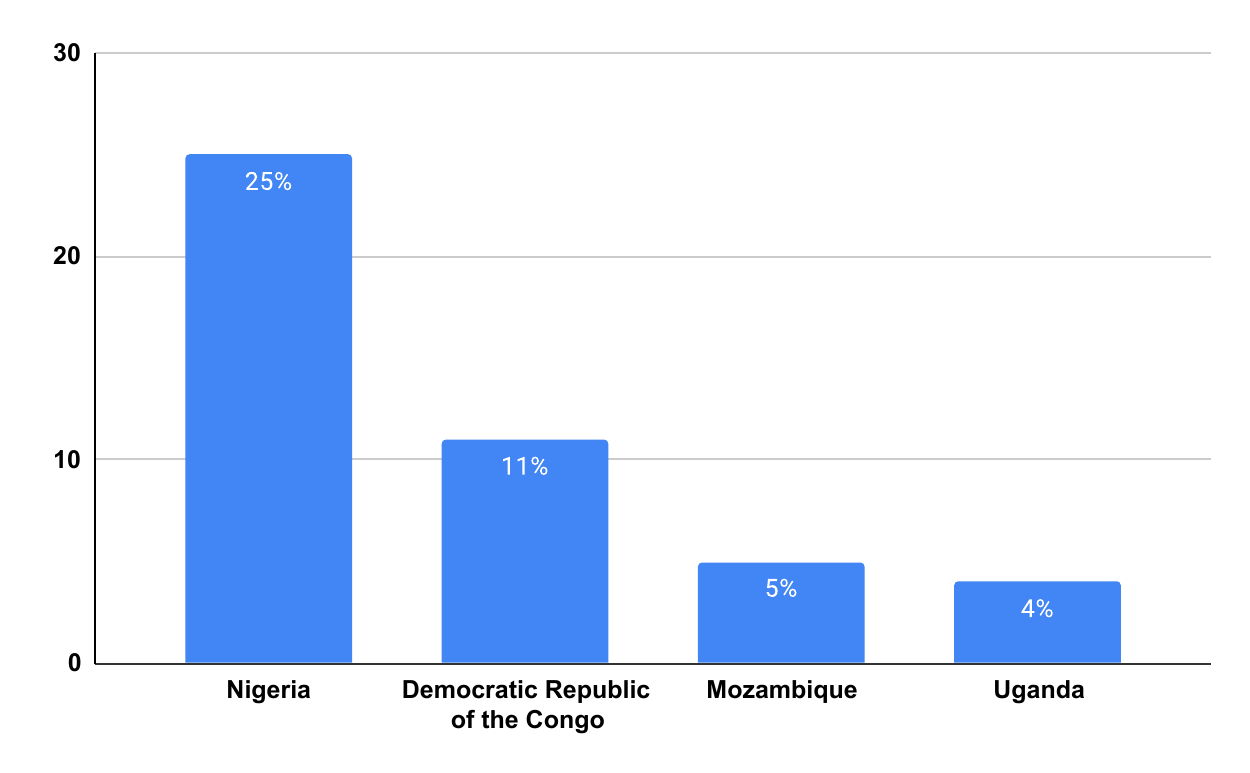} 
    \caption{In 2017, four countries from Africa accounted for 45\% of all malaria cases worldwide.}
    \label{bar11}
\end{figure}

In the study by Ogueda et al. \cite{ref1}, a variant of the physics-informed neural network (PINN), termed the disease-informed neural network (DINN), was employed as a deep learning model. The SIRD compartmental model was utilized to analyze disease dynamics, incorporating the movement of individuals between cities. The primary objective of this work was to predict various parameters, including the rates of transmission, mortality, and recovery for the selected cities, as well as the rate of movement of individuals between them. In  Schiassi et al. \cite{ref6} work, the principles of physics-informed learning were integrated with the theory of functional connections. This approach was applied to three different compartmental models SIR, SIER, and SIERS to predict key metrics such as the transmission rate, recovery rate, and reproduction number. Ning et al. \cite{ref7} utilized PINN alongside the SEIRD compartmental model to calculate a one-week trajectory based on COVID-19 data for Italy. This study also predicted parameters including the transmission rate, recovery rate, and death rate from February 20 to June 25. In another study by Ning et al. \cite{ref8}, a variant of the PINN known as Euler iteration-augmented physics-informed neural networks was implemented in conjunction with the SIRD compartmental model. This approach aimed to predict various transmission rates, including the transmission rate and death rate, while also providing forecasts for 3-day, 5-day, and 7-day intervals. 

Deng et al. \cite{ref2} integrated an RNN-based LSTM deep learning model with the SIRJD compartmental model to predict the transmission rate of COVID-19 in the United States, forecasting disease trajectories for the next 35 and 42 days. Bousquet et al. \cite{ref3} applied an LSTM-based neural network to the SIRD model, predicting transmission and death rates and generating a 10-week trajectory for France, the UK, Germany, and South Korea. Deng et al. \cite{ref4} used DNN and LSTM with the SIRD model for the Omicron variant, predicting transmission parameters and a 28-day trajectory in cities like Shanghai and Hong Kong. Deng and Wang \cite{a} proposed a deep learning approach combining DNN and LSTM to estimate parameters and predict infections and deaths over 28 days, achieving 98 \% accuracy for infections and 92 \% for deaths. For more deep learning and statistical models, we refer to \cite{b, c, d, e, f, g, h, i, j, ref5}.

In the study by Bhuju et al. \cite{ref9}, the temperature dependence of the transmission rate was analyzed using the SEIR model for humans and the LSEI model for mosquitoes. The authors conducted various mathematical analyses, including the stability of disease-free equilibrium and the existence of endemic equilibrium. Numerical simulations across different temperature scenarios revealed that temperature significantly influences the transmission rate.

Keno et al. \cite{ref10} examined the temperature dependency of the transmission parameter using the SIR model for humans and the SI model for mosquitoes. Their analysis includes assessments of both local and global stability of equilibrium points. The study demonstrated that when the basic reproduction number is less than one, the disease-free equilibrium is both locally and globally asymptotically stable. Additionally, the impact of temperature on transmission dynamics was investigated, reinforcing the conclusion that temperature plays a critical role in disease transmission.

In the study by Proctor et al. \cite{ref27}, Dynamic Mode Decomposition (DMD) was utilized to incorporate control effects and extract low-order models from high-dimensional, complex systems. Alla and Kutz \cite{ref28} implemented DMD to reduce the order of a nonlinear dynamical system. Similarly, Andreuzzi et al. \cite{ref29} extended DMD for forecasting future states of parametric dynamical systems. Watson et al. \cite{ref12} employed a Bayesian time series model in conjunction with random forests to predict the number of cases and deaths using the SIRD compartmental model, conducting a 21-day forecast for three cities: New York, Colorado, and West Virginia. Additional research on mathematical models of malaria can be found in \cite{ref11, ref13, ref14, ref15, ref16, ref17, ref18, ref19}.

Most existing studies on malaria transmission dynamics primarily rely on mathematical modeling. In contrast, this work leverages deep learning methods to analyze the dynamics of malaria transmission. One key advantage of using the neural network approach is that these models are designed to emulate the human brain, allowing them to capture complex patterns in data. This capability makes deep learning particularly well-suited for modeling malaria dynamics.

In this study, a simple Artificial Neural Network (ANN) is employed to predict the trajectories of all five compartments of the model. To estimate parameters related to malaria transmission, ANN, Convolutional Neural Networks (CNN), and Recurrent Neural Networks (RNN) are utilized. The use of RNN is particularly advantageous, as it can forecast future values by storing extensive historical data. Additionally, Dynamic Mode Decomposition (DMD) is applied to assess the risk of the disease; DMD is effective in extracting insights from raw data, which distinguishes it from other deep learning methods. The primary factor influencing disease transmission is the transmission rate, which is analyzed in this work while considering both temperature and altitude simultaneously.

This study is organized as follows: Section 2 formulates the model related to disease transmission. Section 3 presents the problem statement, methodology, and results. Finally, Section 4 offers concluding remarks and outlines future directions.
%%%%%%%%%%%%%%%%%%%%%%%%%%%%%%%%%%
%%%%%%%%%%%%%%%%%%%%%%%%%%%%%%%%%%
\mysection{Model formulation} 
In this section, we formulate our model, which encompasses both human and mosquito populations, as malaria involves interactions between these two species. This methodology employs compartment models to analyze the dynamics of humans and mosquitoes. The human population is divided into three compartments: susceptible (S), infected (I), and recovered (R). In contrast, the mosquito population is divided into two compartments: susceptible (S) and infected (I).

Given the two-species nature of this system, the population of one species is influenced by the other. Specifically, when infected mosquitoes come into contact with susceptible humans, the number of infected humans increases. Similarly, when infected humans interact with susceptible mosquitoes, the population of infected mosquitoes rises. Figure \ref{2} illustrates the interactions between the human and mosquito populations.\\
 
\begin{figure}[H]
    \centering
    \includegraphics[scale=0.30]{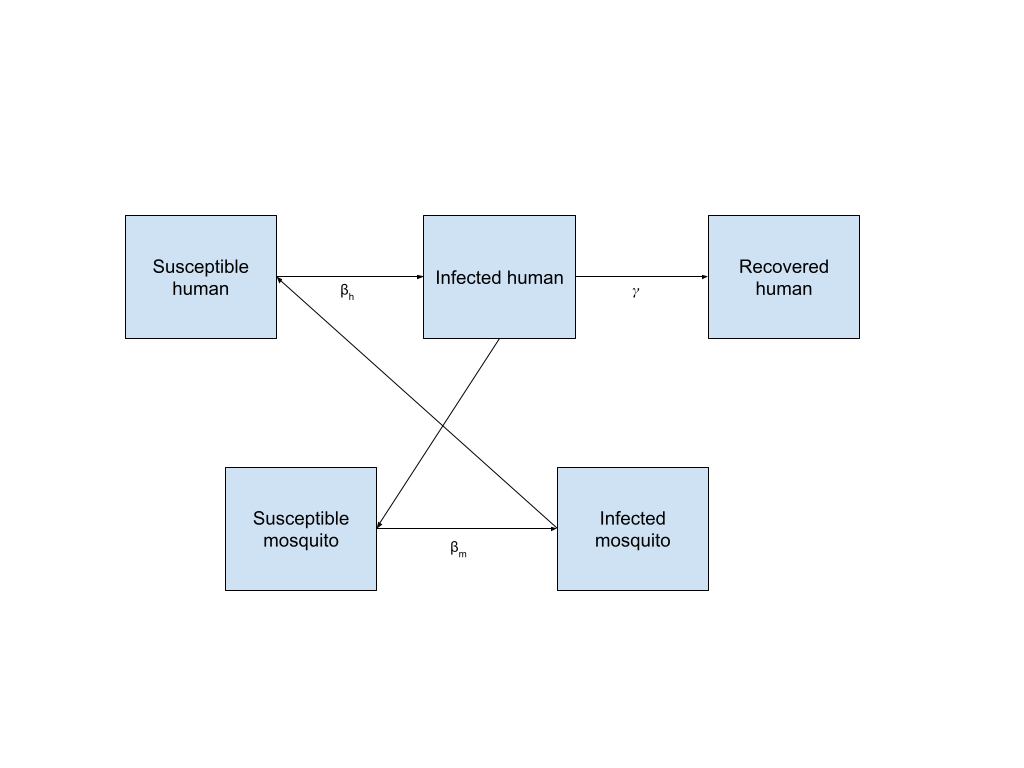} % Replace example.png with your actual image file
    \caption{Schematic diagram of SIR-SI system }
    \label{2}
\end{figure}

For the human population, we consider the standard SIR model, and for the mosquito population, we consider the SI model, the following is the system of differential equations.\\
\begin{eqnarray} \nonumber
\frac{dS_{h}}{dt} &=& \Gamma_{h}N_{h}-\frac{\beta_{h}S_{h}I_{m}}{N_{h}}- \mu_{h}S_{h}  \\ \nonumber 
\frac{dI_{h}}{dt} &=& \frac{\beta_{h}S_{h}I_{m}}{N_{h}} - (\gamma_{h} + \mu_{h})I_{h}\\ \label{1.1} 
\frac{dR_{h}}{dt} &=& \gamma_{h} I_{h} -\mu_{h} R_{h}\\ \nonumber
\frac{dS_{m}}{dt} &=& \Gamma_{m}N_{m}-\frac{\beta_{m}S_{m}I_{h}}{N_{m}}- \mu_{m}S_{m} \\ \nonumber
\frac{dI_{m}}{dt} &=& \frac{\beta_{m}S_{m}I_{h}}{N_{m}} -  \mu_{m}I_{m}. \nonumber
\end{eqnarray}

The model used in this work is the SIR-SI model. The SIR model is used for the human population and the SI model is used for the mosquito population. Table ~\ref{3} and ~\ref{4} explain the compartments and the parameters used. \\
 
\begin{table}[H]
    \centering
    \caption{Interpretation of compartments}
    \vspace{0.1cm}
    \begin{tabular}{|c|p{10cm}|} % Adjust the width of the second column (p{10cm}) as needed
        \hline
        Compartment Symbol & Explanation \\
        \hline
        $S_{h}$  &  It represents the section of the human population who are susceptible to malaria \\
        \hline
        $I_{h}$ &  It represents the section of humans who are infected with malaria  \\
        \hline
        $R_{h}$ & It represents the section of humans who have recovered from malaria \\
        \hline 
        $S_{m}$ & It represents the section of mosquitoes that are susceptible to the malaria-causing parasite \\
        \hline 
        $I_{m}$ & It represents the section of mosquitoes that are infected with the malaria-causing parasite \\
        \hline 
    \end{tabular}
    \label{3}
\end{table}

\begin{table}[h]
    \centering
    \caption{Parameters description} \vspace{0.2 cm}
    \begin{tabular}{|c|c|}
        \hline
        Compartment Symbol & Explanation  \\
        \hline
        $\Gamma_{h} $  &  Birth rate of the humans \\
        \hline
        $N_{h}$ &  Total human population  \\
        \hline
        $\beta_{h}$ & Transmission of malaria in humans \\
        \hline 
        $\mu_{h}$ & mortality rate of humans due to malaria \\
        \hline 
        $\gamma_{h}$ & recovery rate of humans \\
        \hline 
        $\Gamma_{m}$ & Birth rate of mosquitoes\\
        \hline 
        $\beta_{m}$ & transmission of malaria in mosquitoes \\
        \hline 
        $\mu_{m}$ & mortality rate of mosquitoes \\
        \hline 
        $N_{m}$ & Total population of mosquitoes \\
        \hline 
    \end{tabular}
    
    \label{4}
\end{table}

\subsection{Disease-free steady state analysis}
In this section, we perform the steady state analysis. Steady-state solutions play an important role when analytical solution is not known and we want to study the qualitative properties of solutions. 

We define the basic reproduction number  for the model (\ref{1.1}) as \\
$$R_{0} = \frac{\beta_{h}\beta_{m}\Gamma_{h}\Gamma_{m}}{\mu_{h}\mu_{m}^{2}(\mu_{h}+\gamma_{h})}.$$
Observe that the basic reproduction number depends on recruitment rates, infection rates, recovery rates, and mortality rates.\\
\begin{Theorem} 
    If $R_{0}<1$, the disease-free steady state is locally stable. 
\end{Theorem}
\begin{proof} 
For the analysis of the disease-free steady state, we need to equate the infected and recovered population of both species to zero and thus we get
$S_{h} = \frac{\Gamma_{h}}{\mu_{h}}$,
$I_{h} = 0$,
$R_{h} = 0$,
$S_{m} = \frac{\Gamma_{m}}{\mu_{m}}$,
$I_{m} = 0.$\\

Dividing the first three equations $N_{h}$ and the last two equations with $N_{m}$ in (\ref{1.1}), we get
\begin{eqnarray} \nonumber
\frac{dS_{h}}{dt} &=& \Gamma_{h} - m\beta_{h}S_{h}I_{m} - \mu_{h}S_{h} \\ \nonumber
\frac{dI_{h}}{dt} &=& m\beta_{h}S_{h}I_{m} -(\gamma_{h} - \mu_{h})I_{h} \\ \label{1.2}
\frac{dR_{h}}{dt} &=& \gamma_{h}I_{h} - \mu_{h}R_{h} \\ \nonumber
\frac{dS_{m}}{dt} &=& \Gamma_{m} - \frac{\beta_{m}S_{m}I_{h}}{m} - \mu_{m}S_{m} \\ \nonumber
\frac{dI_{m}}{dt} &=& \frac{\beta_{m}S_{m}I_{h}}{m} - \mu_{m}I_{m}.
\end{eqnarray}
Now our task is to linearize the system (\ref{1.2}) around the disease-free steady state. After linearization, we obtain the following system 
\begin{eqnarray*}
\frac{dS_{h}}{dt} &=& \Gamma_{h} - m\beta_{h}\frac{\Gamma_{h}}{\mu_{h}}I_{m} - \mu_{h}S_{h} \\
\frac{dI_{h}}{dt} &=& m\beta_{h}\frac{\Gamma_{h}}{\mu_{h}}I_{m} - (\gamma_{h} + \mu_{h})I_{h}\\
\frac{dR_{h}}{dt} &=& \gamma_{h}I_{h} - \mu_{h}R_{h}\\
\frac{dS_{m}}{dt} &=&\Gamma_{m} - \beta_{m}\frac{\Gamma_{m}}{m\mu_{m}}I_{h} - \mu_{m}S_{m}  \\
\frac{dI_{m}}{dt} &=& \frac{\Gamma_{m}\beta_{m}I_{h}}{m\mu_{m}} - \mu_{m}I_{m}.
\end{eqnarray*}
Constructing the Jacobian matrix, we get 
\[ 
J = 
\begin{bmatrix}
    -\mu_{h} & 0 & 0 & 0 & -\frac{m\beta_{h}\Gamma_{h}}{\mu_{h}} \\
    0 & -(\gamma_{h}+ \mu_{h}) & 0 & 0 &  \frac{m\beta_{h}\Gamma_{h}}{\mu_{h}}\\
    0 & \gamma_{h} & -\mu_{h} & 0 & 0 \\
    0 & \frac{-\beta_{m}\Gamma_{m}}{m\mu_{m}} & 0 & -\mu_{m} & 0 \\
    0 & \frac{\beta_{m}\Gamma_{m}}{m\mu_{m}}& 0 & 0 & -\mu_{m} \\
\end{bmatrix}.
\]

The characteristic polynomial of the above matrix is 
$$f(\lambda) = -(\lambda + \mu_{h})(\lambda + \mu_{h})(\lambda + \mu_{m}) [(\lambda+ \mu_{m})(\lambda  + \mu_{h} + \gamma_{h}) - \frac{\beta_{h}\beta_{m}\Gamma_{h}\Gamma_{m}}{\mu_{m}\mu_{h}}].$$

For the system to be stable, all the eigenvalues must have negative real parts. It is clear that three of the eigenvalues are negative; we need to check the roots of the quadratic polynomial for the remaining two eigenvalues. By solving the quadratic equation, we get the roots as \\
 
$$\frac{-(\gamma + \mu_{h}+\mu_{m}) + \sqrt{(\mu_{m}+ \mu_{h} + \gamma_{h})^{2} - 4(\mu_{m}(\mu_{h}+ \gamma_{h}) - \frac{\beta_{h}\beta_{m}\Gamma_{h}\Gamma_{m}}{\mu_{h}\mu_{m}}) }}{2}$$

$$\frac{-(\gamma + \mu_{h}+\mu_{m}) - \sqrt{(\mu_{m}+ \mu_{h} + \gamma_{h})^{2} - 4(\mu_{m}(\mu_{h}+ \gamma_{h}) - \frac{\beta_{h}\beta_{m}\Gamma_{h}\Gamma_{m}}{\mu_{h}\mu_{m}}) }}{2}.$$
 
If we observe, the second root is always negative, and thus we need to find out the condition for which the first root is negative and that is 
 
$$\mu_{m}( \mu_{h} + \gamma_{h}) - \frac{\beta_{m}\beta_{h}\Gamma_{m}\Gamma_{h}}{\mu_{h}\mu_{m}}> 0. $$
 
Rearranging the above inequality, we get the expression \\
 
$$\frac{\beta_{h}\beta_{m}\Gamma_{h}\Gamma_{m}}{\mu_{h}\mu_{m}^{2}(\mu_{h}+\gamma_{h})}< 1.$$
So, if $R_{0}<1,$ it will ensure us that the disease-free steady state is stable. 
\end{proof}

\subsection{Endemic steady state analysis }
In this subsection, we aim to study the stability of the endemic steady-state. \\
Given system of equations is 
\begin{eqnarray*}
\frac{dS_{h}}{dt} &=& \Gamma_{h} - m\beta_{h}\frac{\Gamma_{h}}{\mu_{h}}I_{m} - \mu_{h}S_{h}\\
\frac{dI_{h}}{dt} &=& m\beta_{h}\frac{\Gamma_{h}}{\mu_{h}}I_{m} - (\gamma_{h} + \mu_{h})I_{h}\\
\frac{dR_{h}}{dt} &=& \gamma_{h}I_{h} - \mu_{h}R_{h}\\
\frac{dS_{m}}{dt} &=& \Gamma_{m} - \beta_{m}\frac{\Gamma_{m}}{m\mu_{m}}I_{h} - \mu_{m}S_{m}  \\
\frac{dI_{m}}{dt} &=& \frac{\Gamma_{m}\beta_{m}I_{h}}{m\mu_{m}} - \mu_{m}I_{m}.
\end{eqnarray*}
 
Before proceeding with the next theorem, it is essential to define the following quantities \\
 
 $b =\mu_{h} + \frac{\beta_{h}I_{m}^{*}}{N_{h}}+ \mu_{h} + \gamma + \frac{\beta_{h}I_{h}^{*}}{N_{m}} + \mu_{m}$\\
  
$c =(\mu_{h} + \frac{\beta_{h}I_{m}^{*}}{N_{h}})(\mu_{h} + \gamma) + (\mu_{h} + \gamma)(\frac{\beta_{h}I_{h}^{*}}{N_{m}} + \mu_{m})+(\frac{\beta_{h}I_{h}^{*}}{N_{m}} + \mu_{m})(\mu_{h} + \frac{\beta_{h}I_{m}^{*}}{N_{h}})- \frac{\beta_{m}\beta_{h}S_{m}^{*}S_{h}^{*}}{N_{m}N_{h}}$\\
 
$d = \left(\mu_{h} + \frac{\beta_{h}I_{m}^{*}}{N_{h}}\right)\left[(\mu_{h} + \gamma)(\frac{\beta_{h}I_{h}^{*}}{N_{m}} + \mu_{m})-\frac{\beta_{m}\beta_{h}S_{m}^{*}S_{h}^{*}}{N_{m}N_{h}} + \frac{\beta_{m}\beta_{h}^{2}S_{m}^{*}S_{h}^{*}I_{m}^{*}}{N_{m}N_{h}^{*}}\right],$\\
 
where  $S^{*}_{m},S^{*}_{h},I^{*}_{m},I^{*}_{h}$ are the non trivial equilibrium solutions. 

\begin{Theorem}
    If $R_{0}>1$ and $b, c, d, bc-d >0$, the endemic steady state is locally stable. 
\end{Theorem}

\begin{proof} 
The equilibrium points will be obtained by equating all of the above time derivatives to zero and by solving them for non-trivial solutions, we obtain the solutions as 
\begin{eqnarray*}
S_{h}^{*} &=& \frac{N_{h}(\gamma + \mu_{h})(\mu_{m}N_{m}+ \frac{\Gamma_{h}N_{h}\beta_{m}}{\gamma + \mu_{h}})}{\beta_{m}(\mu_{h}N_{h}+ \frac{\Gamma_{m}N_{m}\beta_{h}}{\mu_{m}})} \\
I_{h}^{*} &=& \frac{N_{h}N_{m}\mu_{m}\mu_{h}(R_{0}-1)}{\beta_{m}(\mu_{h}N_{h}+ \frac{\Gamma_{m}N_{m}\beta_{h}}{\mu_{m}})} \\
R_{h}^{*} &=& \frac{N_{h}N_{m}\mu_{m}\gamma_{h}(R_{0}-1)}{\beta_{m}(\mu_{h}N_{h}+ \frac{\Gamma_{m}N_{m}\beta_{h}}{\mu_{m}})} \\
 S_{m}^{*} &=& \frac{N_{m}\mu_{m}(\mu_{h}N_{h}+ \frac{\Gamma_{m}N_{m}\beta_{h}}{\mu_{m}})}{\beta_{h}(\mu_{m}N_{m}+ \frac{\Gamma_{h}N_{h}\beta_{m}}{\gamma + \mu_{h}})}\\ 
I_{m}^{*} &=& \frac{N_{h}N_{m}\mu_{m}\mu_{h}(R_{0}-1)}{\beta_{h}(\mu_{m}N_{m}+ \frac{\Gamma_{h}N_{h}\beta_{m}}{\gamma + \mu_{h}})}. 
\end{eqnarray*}
Linearizing  the model (\ref{1.2})  around the endemic equilibrium point, we get the following system
\begin{eqnarray*}
\frac{dS_{h}}{dt} &=& \Gamma_{h}N_{h}-\frac{\beta_{h}}{N_{h}}(S_{h}^{*}I_{m} + I_{m}^{*}S_{h}- S_{h}^{*}I_{m}^{*})- \mu_{h}S_{h}\\
\frac{dI_{h}}{dt} &=&\frac{\beta_{h}}{N_{h}}(S_{h}^{*}I_{m} + I_{m}^{*}S_{h}-S_{h}^{*}I_{m}^{*})- \gamma I_{h} -\mu_{h}I_{h} \\
\frac{dR_{h}}{dt} &=& \gamma I_{h} - \mu_{h}R_{h} \\
\frac{dS_{m}}{dt} &=& \gamma_{m}N_{m}-\frac{\beta_{m}}{N_{m}}(S_{m}^{*}I_{h}+ S_{m} I_{h}*) - \mu_{m}S_{m} \\
\frac{dI_{m}}{dt} &=& \frac{\beta_{m}}{N_{m}}(S_{m}^{*}I_{h}+ S_{m} I_{h}*) - \mu_{m}I_{m}.
\end{eqnarray*}
Writing the Jacobian matrix, we get \\
\[
J = 
\begin{bmatrix}
     -\frac{\beta_{h}I_{m}^{*}}{N_{h}}-\mu_{h} & 0 & 0 & 0 & -\frac{\beta_{h}S_{h}^{*}}{N_{h}}\\
    \frac{\beta_{h}I_{m}^{*}}{N_{h}} & -(\gamma_{h}+ \mu_{h}) & 0 & 0 & \frac{\beta_{h}S_{h}^{*}}{N_{h}} \\
    0 & \gamma_{h} & -\mu_{h} & 0 & 0 \\
    0 & -\frac{\beta_{h}S_{m}^{*}}{N_{m}} & 0 & -\frac{\beta_{h}S_{m}^{*}}{N_{m}}-\mu_{m} & 0 \\
    0 &\frac{\beta_{h}S_{m}^{*}}{N_{m}} & 0 &\frac{\beta_{h}S_{m}^{*}}{N_{m}}  & -\mu_{m} \\
\end{bmatrix}
\]
The characteristic polynomial of the above matrix is \\
\[
f(\lambda) = -\left(\lambda + \mu_{h}\right)(\lambda + \mu_{m}) \left( \left[\lambda + \mu_{h} + \frac{\beta_{h}I_{m}^{*}}{N_{h}} \right]\left[ \left(\lambda + \mu_{h}+\gamma\right) \left(\lambda + \frac{\beta_{m}I_{h}^{*}}{N_{m}}+\mu_{m} \right)-\frac{\beta_{m}\beta_{h}S_{m}^{*}S_{h}^{*}}{N_{h}N_{m}} \right] + \frac{\beta_{h}^{2}\beta_{m}S_{h}^{*}I_{m}^{*}S_{m}^{*}}{N_{h}^{2}N_{m}} \right).\]

From the above equation, we can see that we have two linear factors and a cubic factor. To analyze cubic factor, let us state the following lemma:
\begin{Lemma}
Let $f(x) = ax^3 + bx^2 + cx + d$ be a cubic polynomial. For $f(x)$ to have all negative roots or complex roots with negative real parts, the following conditions are necessary:\\
$$a  > 0, b  > 0, c  > 0, d  >0, bc - ad  >0.$$
\end{Lemma} 
Now by using the above-stated lemma, we can arrive at our required result.
\end{proof} 
\subsection{Numerical validation of the stability theorems}
In the previous section, we observed that the disease-free steady state is achieved when $R_0 < 1$. Accordingly, we selected $R_0$ such that $R_0 < 1$, and as shown in Figure \ref{less than 1}, the infected populations of humans and mosquitoes diminish over time.  
% and the other with $R_{0}>1$ which can be found in Figure \ref{greater than 1}. From both of the graphs, we can observe that when $R_{0}>1$ the infected human population diverges to infinity thus proving the instability of the model.  
\begin{figure}[H]
    \centering
    \includegraphics[scale = 0.45]{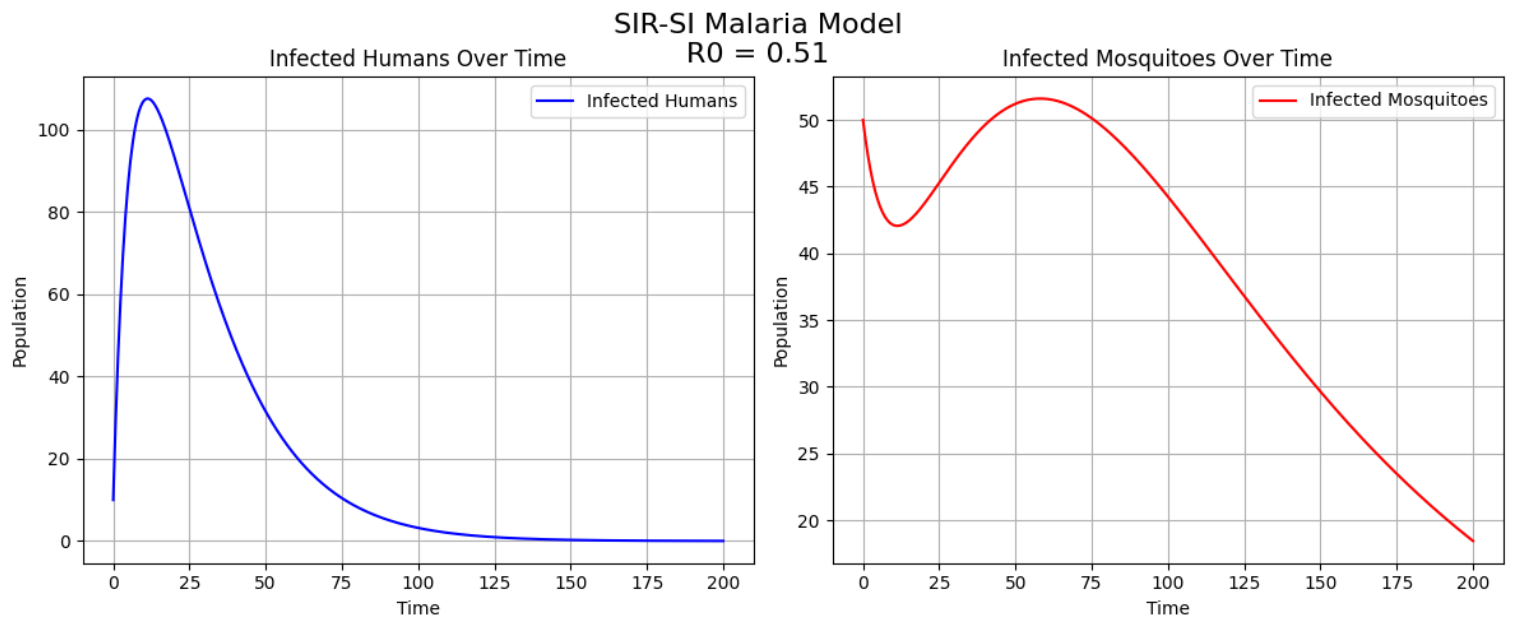}
    \caption{When reproduction number is less than 1}
    \label{less than 1}
\end{figure}

% \begin{figure}[H]
%     \centering
%     \includegraphics[scale = 0.45]{greater than 1.png}
%     \caption{When reproduction number is greater than 1}
%     \label{greater than 1}
% \end{figure} 
\subsection{Temperature and altitude dependence of the transmission rate}
The transmission rate of the malaria disease is heavily dependent on the altitude and temperature, and thus to consider both of the factors we write
\begin{equation} \label{3.1}
\beta(T,h) = \beta_{0}e^{-\frac{(T-25)^{2}}{\eta^{2}}}e^{-\frac{h^{2}}{\xi^{2}}}(1-e^{-\frac{h^{2}}{\xi^{2}}}),
\end{equation}
  
 where $\beta_{0}$ is the transmission constant of the region, $T$ is the temperature of the region, $h$ is the altitude, $\eta$ and $\xi$ are the constants associated with the region's temperature and height respectively.\\
  
 Here $e^{-\frac{(T-25)^{2}}{\eta^{2}}}$ is used to model the temperature and
 $e^{-\frac{h^{2}}{\xi^{2}}}(1-e^{-\frac{h^{2}}{\xi^{2}}})$ is used to model the height. Malaria transmission will be very minimal when the temperature is either extremely high or it is extremely low and thus to model this variation, the Gaussian function is used and the reason for shifting it by $25$ is because the optimum temperature for mosquito's existence and malaria transmission is $25\,^\circ\mathrm{C}$. Also, malaria transmission is completely zero when the altitude is zero since there won't be any mosquitoes in the sea and in the same way when the altitude is extremely high again the transmission is completely zero since there are no mosquitoes in the space and thus to model both of these conditions the negative exponential function is used in this manner.\\
 
Assuming the temperature is from $T_1$ to $T_2$ and the height is from $h_1$ to $h_2$, we can write 
\begin{equation} \label{3.2}
\beta_{avg} = \int_{T_{1}}^{T_{2}} \int_{h_{1}}^{h_{2}}e^{-\frac{(T-25)^{2}}{\eta^{2}}}e^{-\frac{h^{2}}{\xi^{2}}}(1-e^{-\frac{h^{2}}{\xi^{2}}})dh  dT.
\end{equation}
Here the effect of transmission rate is studied by changing the temperature values for a fixed height. The parameters considered are the following:\\
$\beta_{0} = 10 $,
 and for human we have taken $\eta=200 ,\xi = 20000$  and for mosquitoes, $\eta = 400 , \xi = 40000$. The different temperature values which we used are $25^\circ$, $30^\circ$, $35^\circ$, $40^\circ$, $45 ^\circ$.\\

The effect of temperature on the transmission rate while maintaining a height of 75 m can be observed in Figures~\ref{5} and~\ref{6}, while the effect at a height of 100 m is shown in Figures~\ref{7} and~\ref{8}, and at a height of 125 m in Figures~\ref{9} and~\ref{10}.\\

\begin{figure}[H]
    \centering
    \begin{subfigure}[b]{0.45\textwidth}
        \centering
        \includegraphics[width=\textwidth]{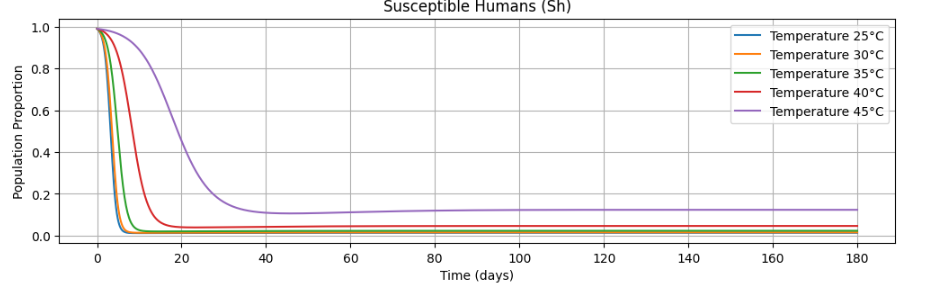}
        \caption{Susceptible human}
        \label{fig:subfiga1}
    \end{subfigure}
    \hfill
    \begin{subfigure}[b]{0.45\textwidth}
        \centering
        \includegraphics[width=\textwidth]{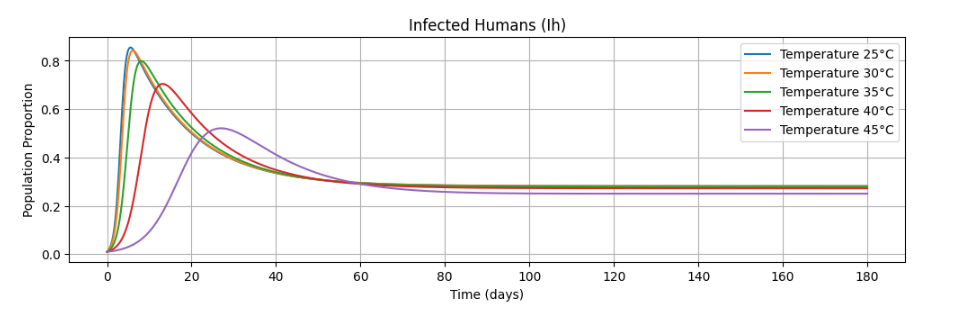}
        \caption{Infected human}
        \label{fig:subfigb2}
    \end{subfigure}
     \caption{Effect on human population when height is 75 m }
    \label{5}
\end{figure}
\begin{figure}[H]
    \centering
    \begin{subfigure}[b]{0.45\textwidth}
        \centering
        \includegraphics[width=\textwidth]{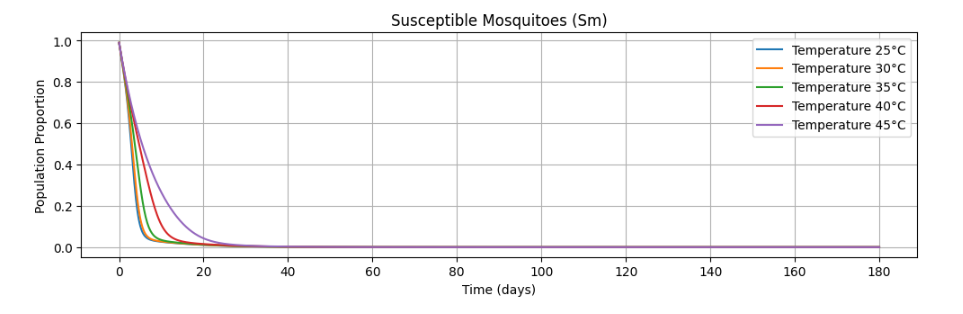}
        \caption{Susceptible mosquito}
        \label{fig:subfiga3}
    \end{subfigure}
    \hfill
    \begin{subfigure}[b]{0.45\textwidth}
        \centering
        \includegraphics[width=\textwidth]{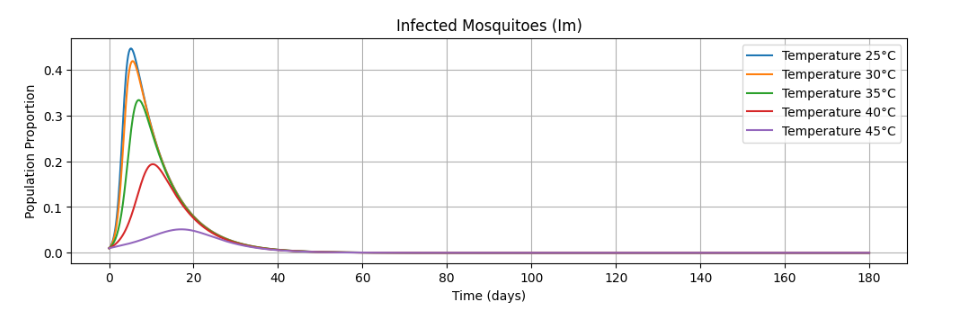}
        \caption{Infected mosquito}
        \label{fig:subfigb4}
    \end{subfigure}
    \caption{Effect on mosquitoes population when height is 75 m}
    \label{6}
\end{figure}

\begin{figure}[H]
    \centering
    \begin{subfigure}[b]{0.45\textwidth}
        \centering
        \includegraphics[width=\textwidth]{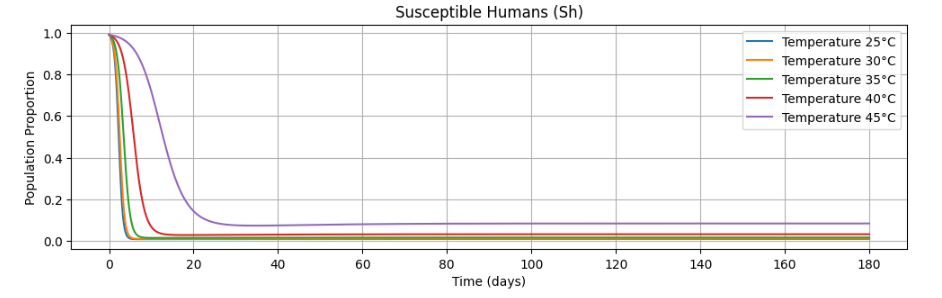}
        \caption{Susceptible human}
        \label{fig:subfiga5}
    \end{subfigure}
    \hfill
    \begin{subfigure}[b]{0.45\textwidth}
        \centering
        \includegraphics[width=\textwidth]{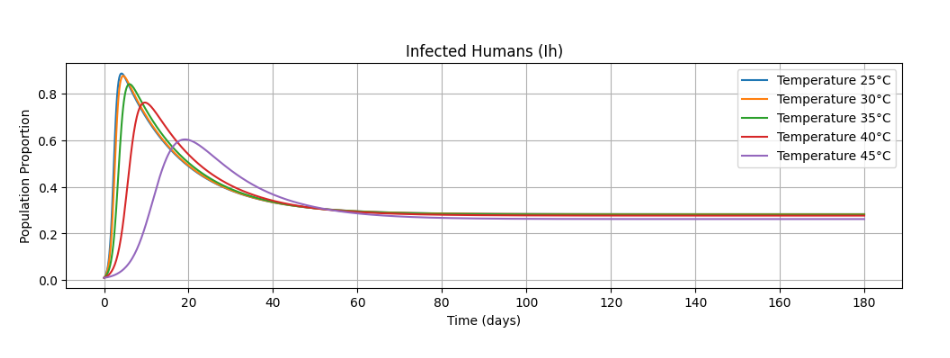}
        \caption{Infected human}
        \label{fig:subfigb6}
    \end{subfigure}
    \caption{Effect on human population when height is 100 m}
    \label{7}
\end{figure}
\begin{figure}[H]
    \centering
    \begin{subfigure}[b]{0.45\textwidth}
        \centering
        \includegraphics[width=\textwidth]{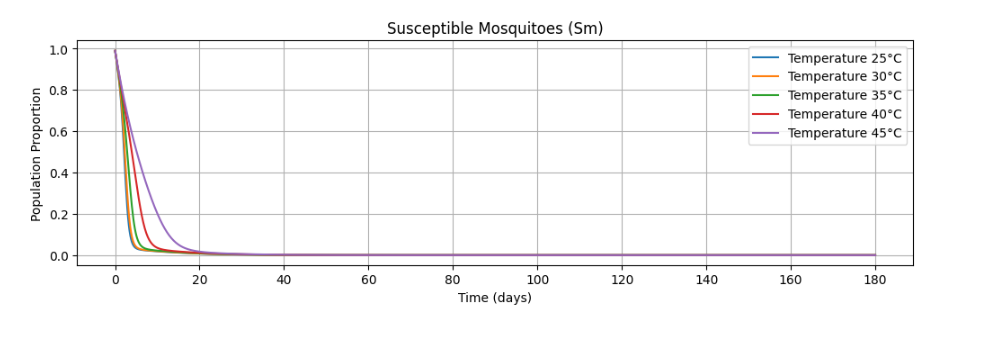}
        \caption{Susceptible mosquito}
        \label{fig:subfiga7}
    \end{subfigure}
    \hfill
    \begin{subfigure}[b]{0.45\textwidth}
        \centering
        \includegraphics[width=\textwidth]{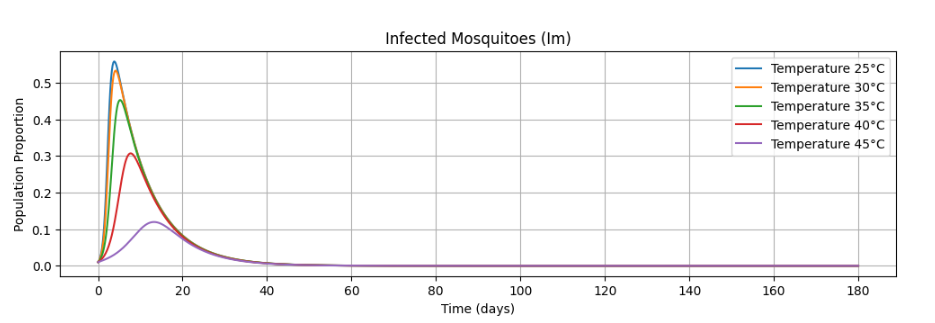}
        \caption{Infected mosquito}
        \label{fig:subfigb8}
    \end{subfigure}
    \caption{Effect on mosquitoes population when height is 100 m}
    \label{8}
\end{figure}
 
\begin{figure}[H]
    \centering
    \begin{subfigure}[b]{0.45\textwidth}
        \centering
        \includegraphics[width=\textwidth]{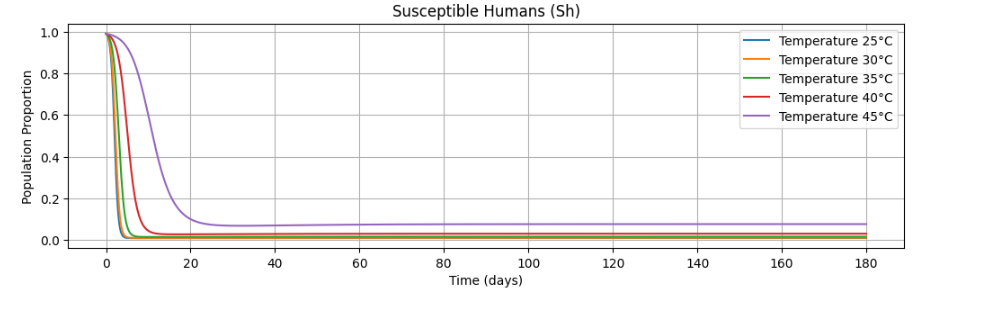}
        \caption{Susceptible human}
        \label{fig:subfiga9}
    \end{subfigure}
    \hfill
    \begin{subfigure}[b]{0.45\textwidth}
        \centering
        \includegraphics[width=\textwidth]{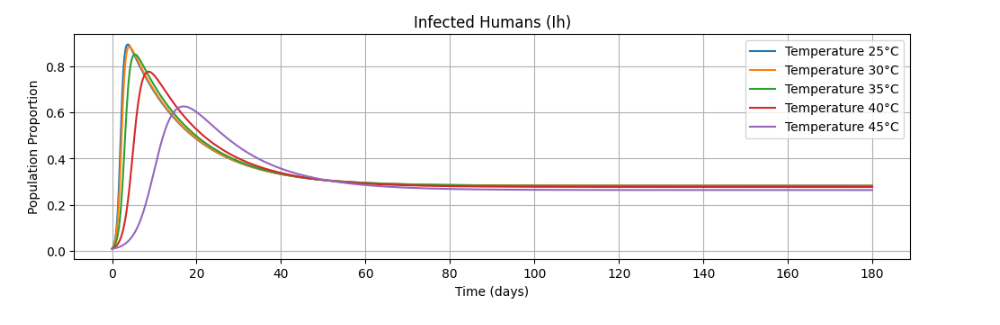}
        \caption{Infected human}
        \label{fig:subfigb10}
    \end{subfigure}
    \caption{Effect on human population when height is 125 m}
    \label{9}
\end{figure}
\begin{figure}[H]
    \centering
    \begin{subfigure}[b]{0.45\textwidth}
        \centering
        \includegraphics[width=\textwidth]{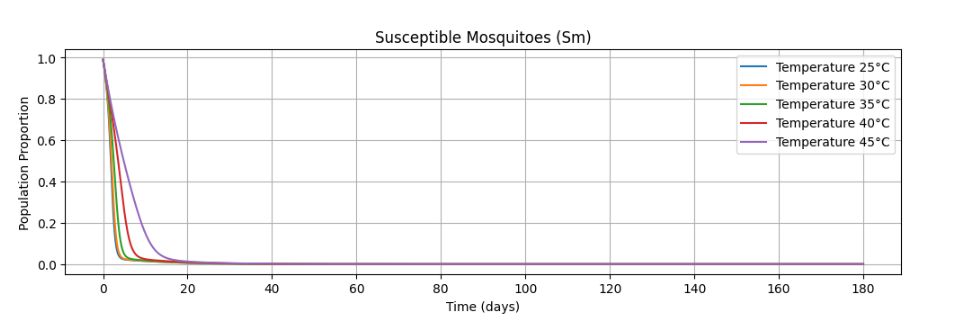}
        \caption{Susceptible mosquito}
        \label{fig:subfiga11}
    \end{subfigure}
    \hfill
    \begin{subfigure}[b]{0.45\textwidth}
        \centering
        \includegraphics[width=\textwidth]{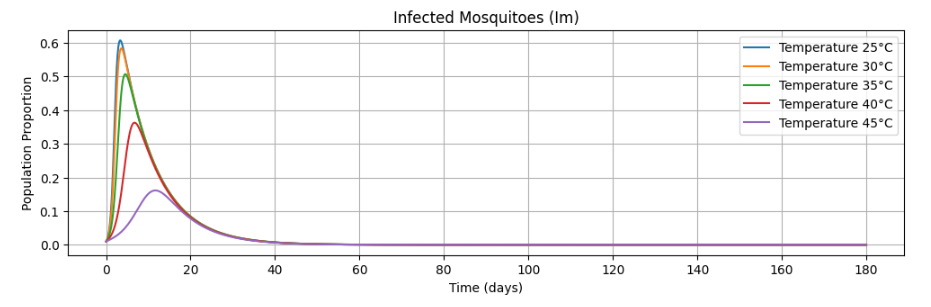}
        \caption{Infected mosquito}
        \label{fig:subfigb12}
    \end{subfigure}
    \caption{Effect on mosquito population when height is 125 m}
    \label{10}
\end{figure}
Thus in this graph, we have plotted the graphs for various temperatures and we can observe that the transmission rate is the maximum for $25^\circ$ which is the least of all the temperatures. Next, we verify if $25^\circ$ is the optimal temperature. We consider the following values for temperature $10^\circ$, $18^\circ$, $25^\circ$.
The effect of temperature on the transmission rate while maintaining a height of 75 m can be observed in Figures~\ref{11} and~\ref{12}, while at a height of 100 m, it can be observed in Figures~\ref{13} and~\ref{14}, and at a height of 125 m in Figures~\ref{15} and~\ref{16}.\\
 
\begin{figure}[H]
    \centering
    \begin{subfigure}[b]{0.45\textwidth}
        \centering
        \includegraphics[width=\textwidth]{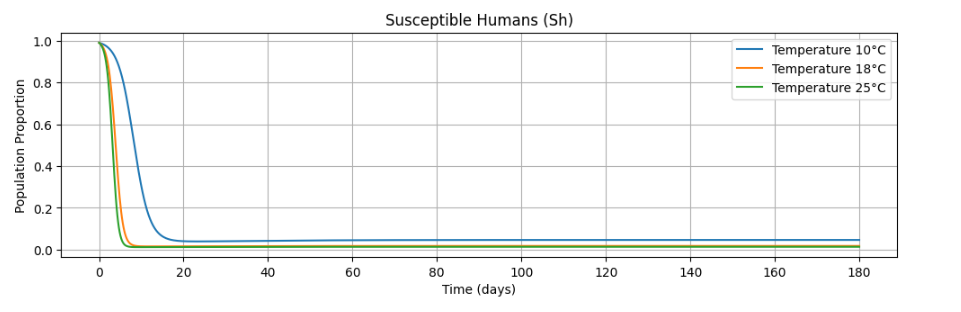}
        \caption{Susceptible human}
        \label{fig:subfiga13}
    \end{subfigure}
    \hfill
    \begin{subfigure}[b]{0.45\textwidth}
        \centering
        \includegraphics[width=\textwidth]{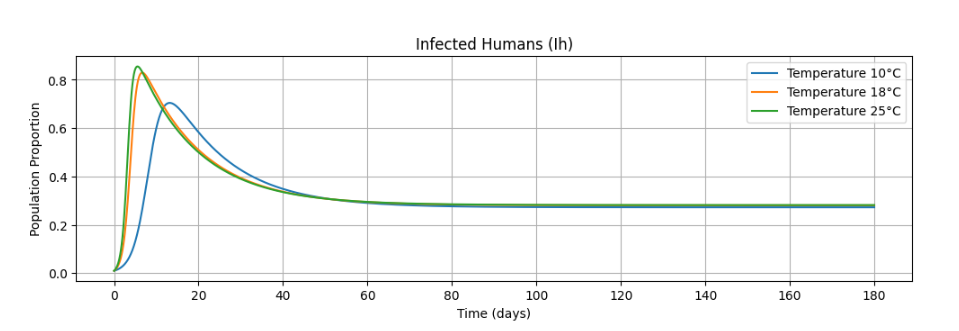}
        \caption{Infected human}
        \label{fig:subfigb14}
    \end{subfigure}
     \caption{Effect on human population when height is 75 m }
    \label{11}
\end{figure}
\begin{figure}[H]
    \centering
    \begin{subfigure}[b]{0.45\textwidth}
        \centering
        \includegraphics[width=\textwidth]{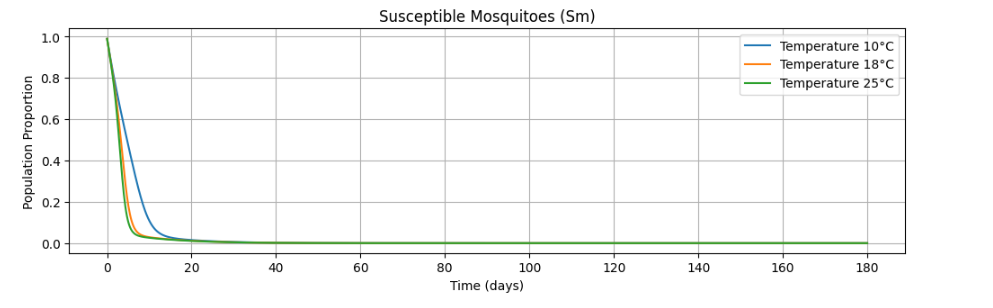}
        \caption{Susceptible mosquito}
        \label{fig:subfiga15}
    \end{subfigure}
    \hfill
    \begin{subfigure}[b]{0.45\textwidth}
        \centering
        \includegraphics[width=\textwidth]{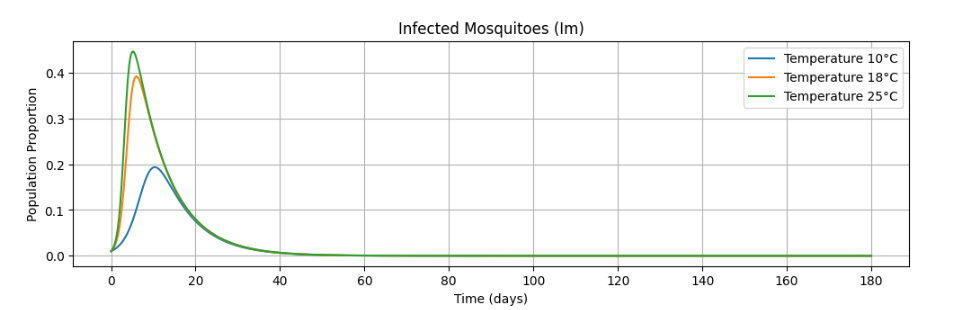}
        \caption{Infected mosquito}
        \label{fig:subfigb16}
    \end{subfigure}
    \caption{Effect on mosquito population when height is 75 m }
    \label{12}
\end{figure}

\begin{figure}[H]
    \centering
    \begin{subfigure}[b]{0.45\textwidth}
        \centering
        \includegraphics[width=\textwidth]{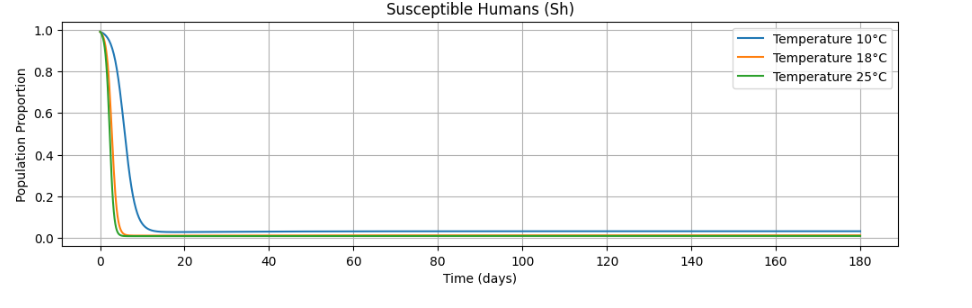}
        \caption{Susceptible human}
        \label{fig:subfiga17}
    \end{subfigure}
    \hfill
    \begin{subfigure}[b]{0.45\textwidth}
        \centering
        \includegraphics[width=\textwidth]{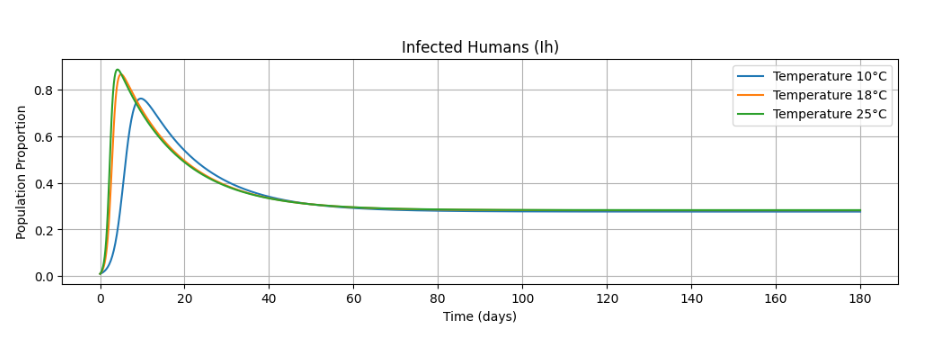}
        \caption{Infected human}
        \label{fig:subfigb18}
    \end{subfigure}
    \caption{Effect on human population when height is 100 m}
    \label{13}
\end{figure}
\begin{figure}[H]
    \centering
    \begin{subfigure}[b]{0.45\textwidth}
        \centering
        \includegraphics[width=\textwidth]{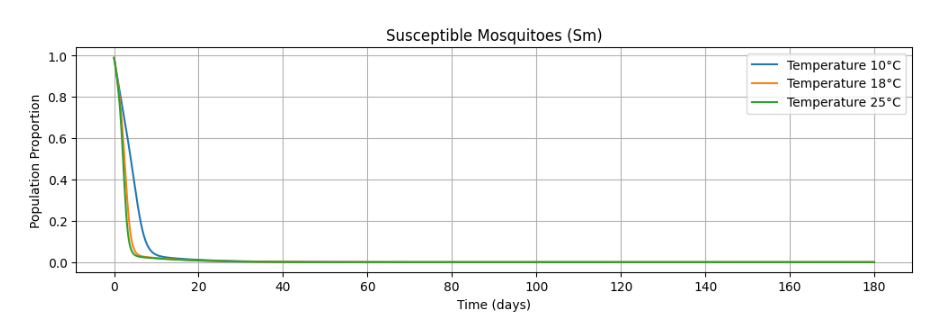}
        \caption{Susceptible mosquito}
        \label{fig:subfiga19}
    \end{subfigure}
    \hfill
    \begin{subfigure}[b]{0.45\textwidth}
        \centering
        \includegraphics[width=\textwidth]{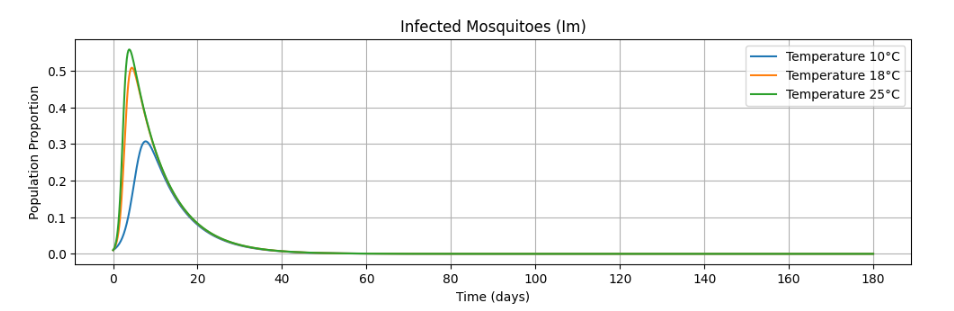}
        \caption{Infected mosquito}
        \label{fig:subfigb20}
    \end{subfigure}
    \caption{Effect on mosquito population when height is 100 m}
    \label{14}
\end{figure}
 
\begin{figure}[H]
    \centering
    \begin{subfigure}[b]{0.45\textwidth}
        \centering
        \includegraphics[width=\textwidth]{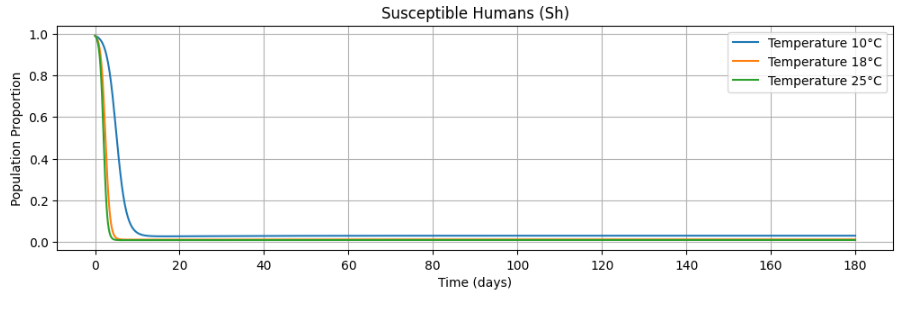}
        \caption{Susceptible human}
        \label{fig:subfiga21}
    \end{subfigure}
    \hfill
    \begin{subfigure}[b]{0.45\textwidth}
        \centering
        \includegraphics[width=\textwidth]{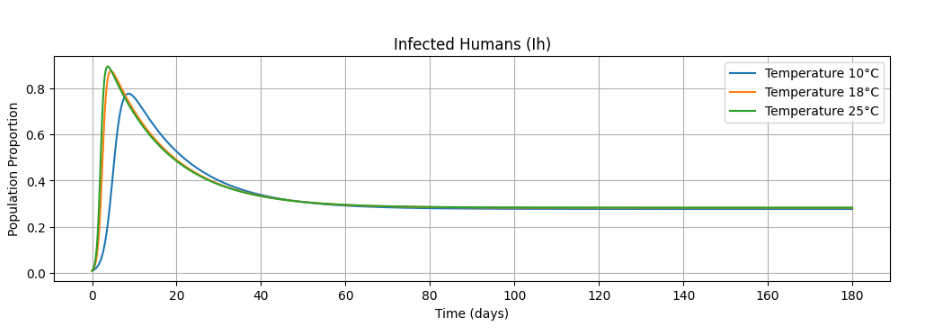}
        \caption{Infected human}
        \label{fig:subfigb22}
    \end{subfigure}
    \caption{Effect on human population  when height is 125 m}
    \label{15}
\end{figure}
\begin{figure}[H]
    \centering
    \begin{subfigure}[b]{0.45\textwidth}
        \centering
        \includegraphics[width=\textwidth]{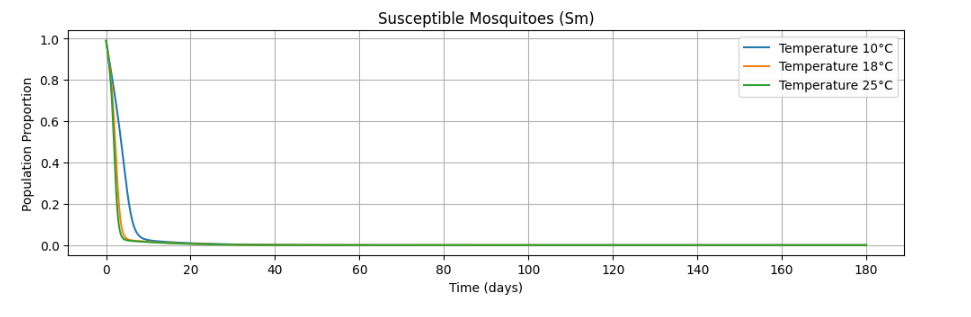}
        \caption{Susceptible mosquito}
        \label{fig:subfiga23}
    \end{subfigure}
    \hfill
    \begin{subfigure}[b]{0.45\textwidth}
        \centering
        \includegraphics[width=\textwidth]{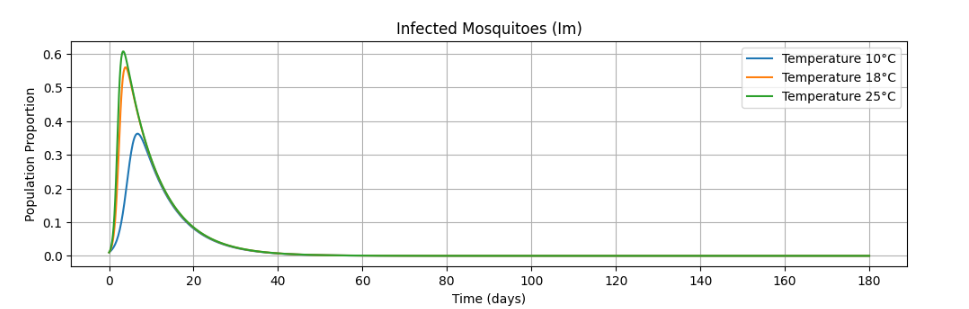}
        \caption{Infected mosquito}
        \label{fig:subfigb24}
    \end{subfigure}
    \caption{Effect on mosquito population when height is 125 m }
    \label{16}
\end{figure}
 
Again, we can observe that the transmission is the highest when the temperature is  $25^\circ$ which is the highest of all temperatures. Thus the threshold value of the temperature is  $25^\circ$ and the transmission rate is the highest only when the temperature is $25^\circ$. The main reason why this happens is that the temperature component is $e^{-\frac{(T-25)^{2}}{\eta^{2}}},$ and this component achieves the maximum value when  $T =25^\circ $.
 
Malaria transmission is also analyzed by keeping the temperature constant and varying the height.
Values of altitude used are $150, 170, 200, 220, ~\text{and}~ 250$ m. The effect of height on the transmission rate can be observed under different temperatures at $28^\circ$ in Figures~\ref{A} and~\ref{B}, at $35^\circ$ in Figures~\ref{C} and~\ref{D}, and at $42^\circ$ in Figures~\ref{E} and~\ref{F}.
 
\begin{figure}[H]
    \centering
    \begin{subfigure}[b]{0.45\textwidth}
        \centering
        \includegraphics[width=\textwidth]{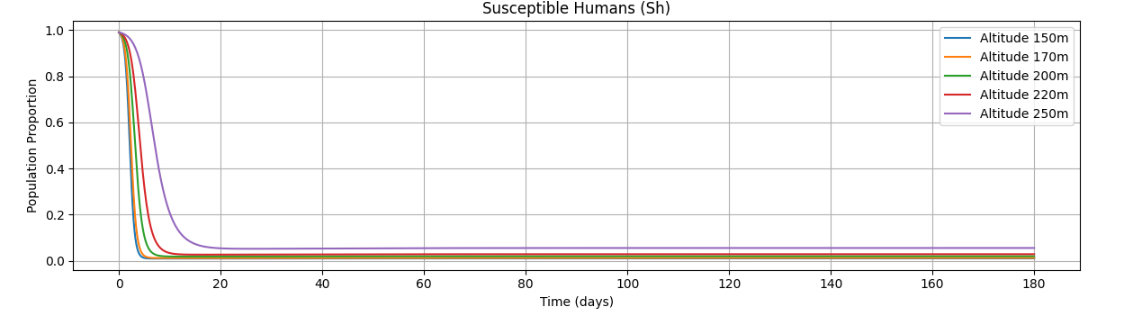}
        \caption{Susceptible humans}
        \label{fig:subfiga25}
    \end{subfigure}
    \hfill
    \begin{subfigure}[b]{0.45\textwidth}
        \centering
        \includegraphics[width=\textwidth]{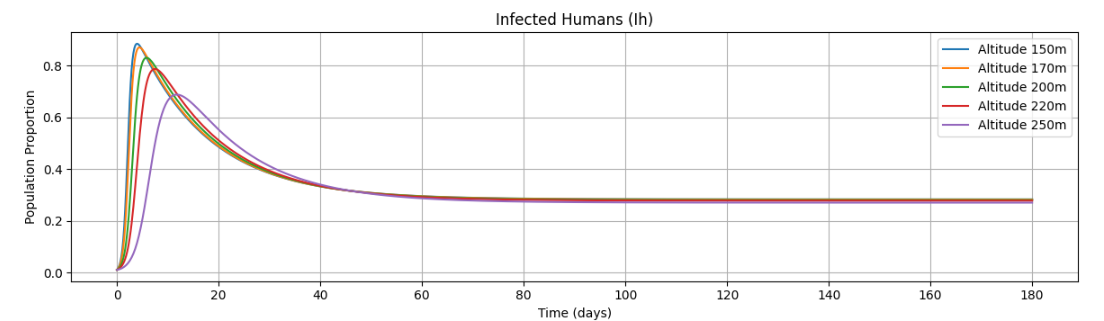}
        \caption{Infected human}
        \label{fig:subfigb26}
    \end{subfigure}
     \caption{Dynamics of human population when  temperature is \( 28 \) degrees Celsius}
    \label{A}
\end{figure}
\begin{figure}[H]
    \centering
    \begin{subfigure}[b]{0.45\textwidth}
        \centering
        \includegraphics[width=\textwidth]{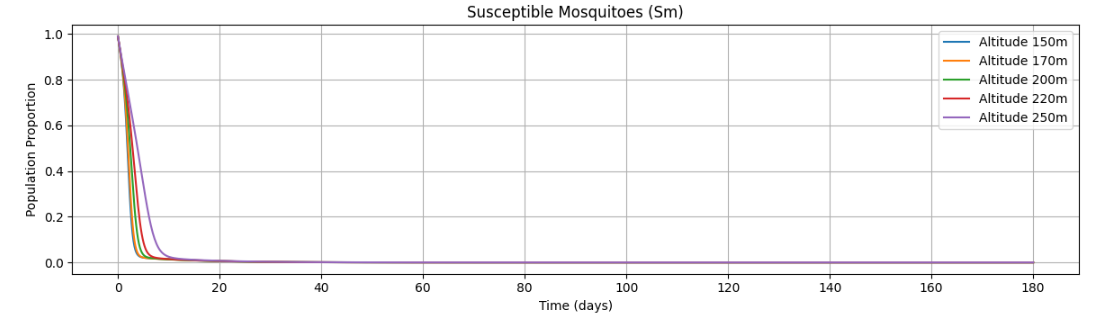}
        \caption{Susceptible mosquito}
        \label{fig:subfiga27}
    \end{subfigure}
    \hfill
    \begin{subfigure}[b]{0.45\textwidth}
        \centering
        \includegraphics[width=\textwidth]{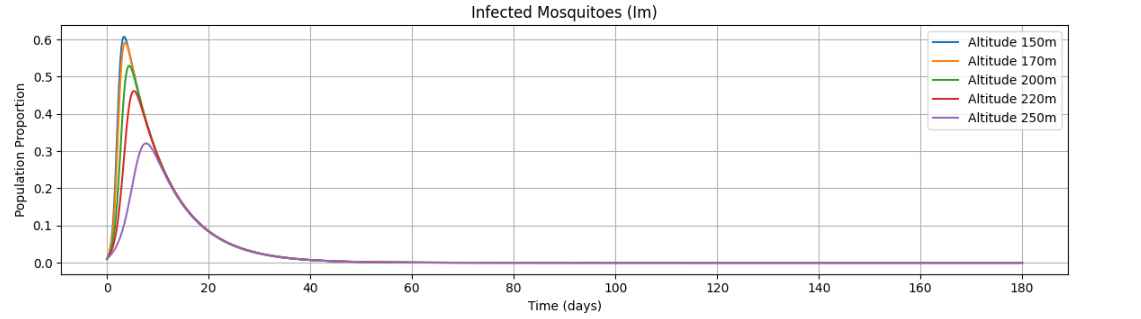}
        \caption{Infected mosquito}
        \label{fig:subfigb28}
    \end{subfigure}
    \caption{Dynamics of mosquito population when temperature is \( 28 \) degrees Celsius}
    \label{B}
\end{figure}
\begin{figure}[H]
    \centering
    \begin{subfigure}[b]{0.45\textwidth}
        \centering
        \includegraphics[width=\textwidth]{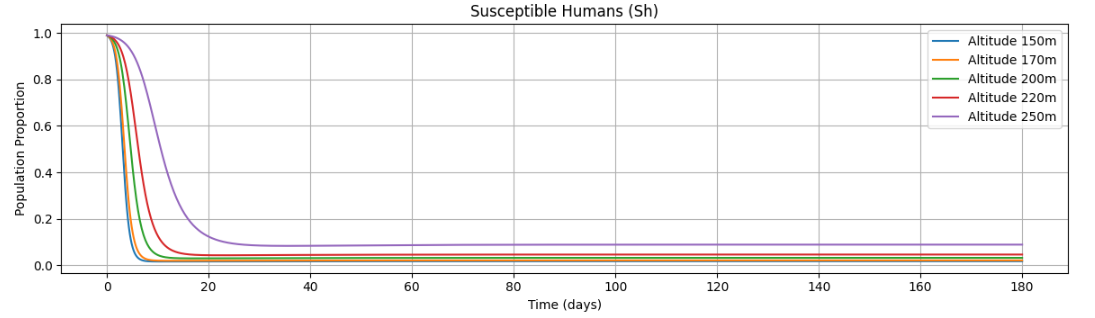}
        \caption{Susceptible human}
        \label{fig:subfiga29}
    \end{subfigure}
    \hfill
    \begin{subfigure}[b]{0.45\textwidth}
        \centering
        \includegraphics[width=\textwidth]{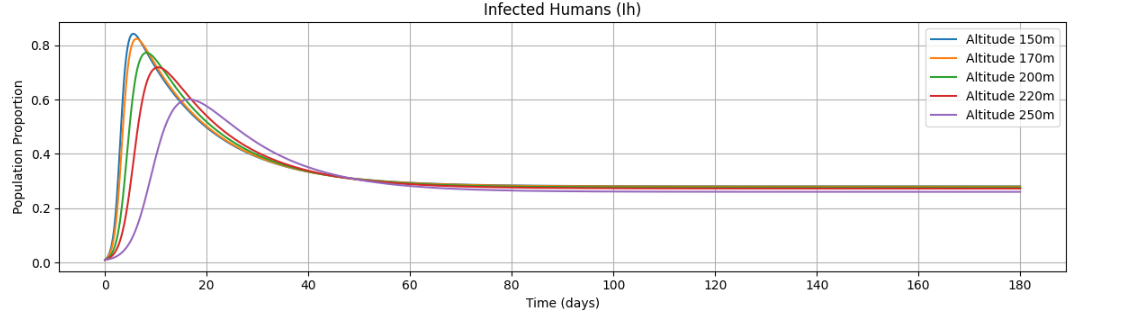}
        \caption{Infected human}
        \label{fig:subfigb30}
    \end{subfigure}
    \caption{ Dynamics of human population when temperature is \( 35\) degrees Celsius}
    \label{C}
\end{figure}
\begin{figure}[H]
    \centering
    \begin{subfigure}[b]{0.45\textwidth}
        \centering
        \includegraphics[width=\textwidth]{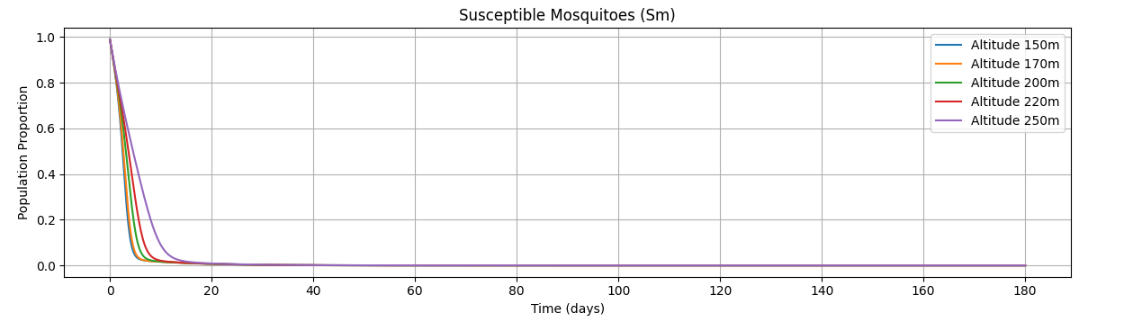}
        \caption{Susceptible mosquito}
        \label{fig:subfiga31}
    \end{subfigure}
    \hfill
    \begin{subfigure}[b]{0.45\textwidth}
        \centering
        \includegraphics[width=\textwidth]{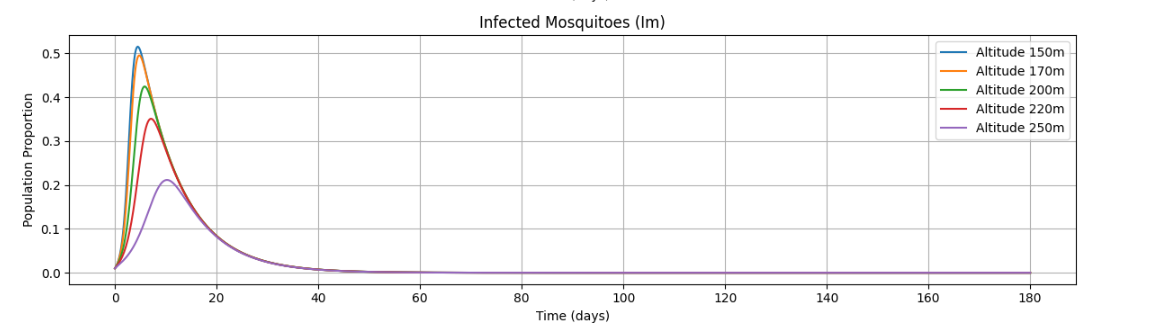}
        \caption{Infected mosquito}
        \label{fig:subfigb32}
    \end{subfigure}
    \caption{ Dynamics of mosquito population when temperature is \( 35\) degrees Celsius}
    \label{D}
\end{figure}
 
\begin{figure}[H]
    \centering
    \begin{subfigure}[b]{0.45\textwidth}
        \centering
        \includegraphics[width=\textwidth]{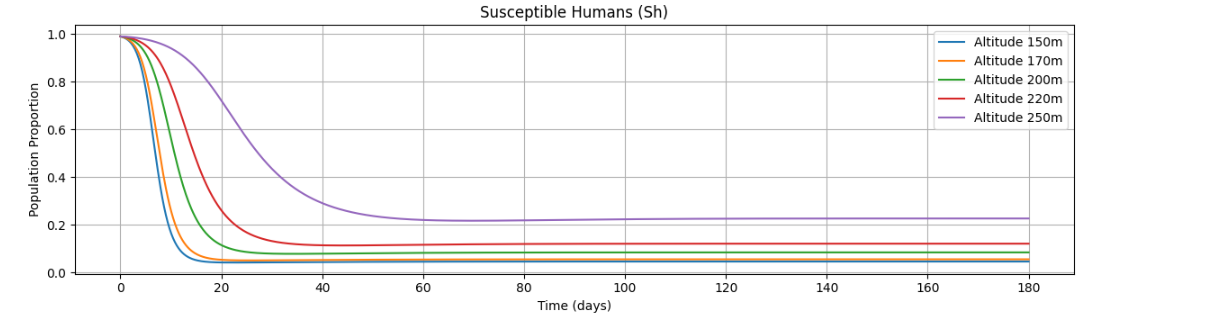}
        \caption{Susceptible human}
        \label{fig:subfiga33}
    \end{subfigure}
    \hfill
    \begin{subfigure}[b]{0.45\textwidth}
        \centering
        \includegraphics[width=\textwidth]{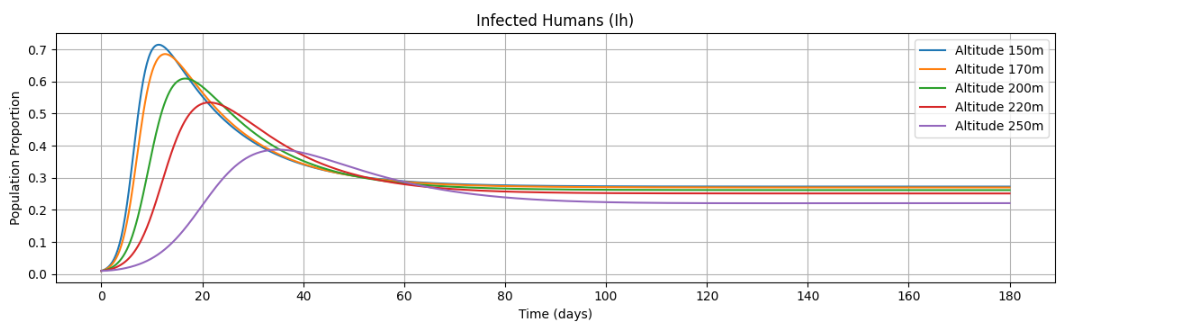}
        \caption{Infected human}
        \label{fig:subfigb34}
    \end{subfigure}
     \caption{ Dynamics of human population when temperature is \( 42\) degrees Celsius}
    \label{E}
\end{figure}
\begin{figure}[H]
    \centering
    \begin{subfigure}[b]{0.45\textwidth}
        \centering
        \includegraphics[width=\textwidth]{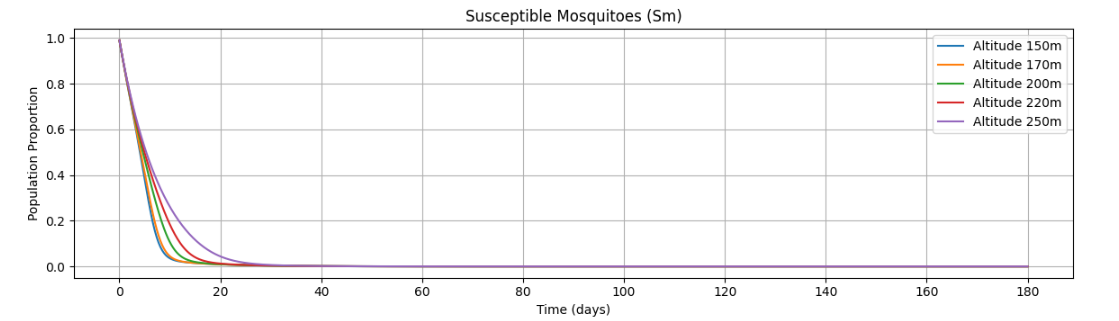}
        \caption{Susceptible mosquito}
        \label{fig:subfiga35}
    \end{subfigure}
    \hfill
    \begin{subfigure}[b]{0.45\textwidth}
        \centering
        \includegraphics[width=\textwidth]{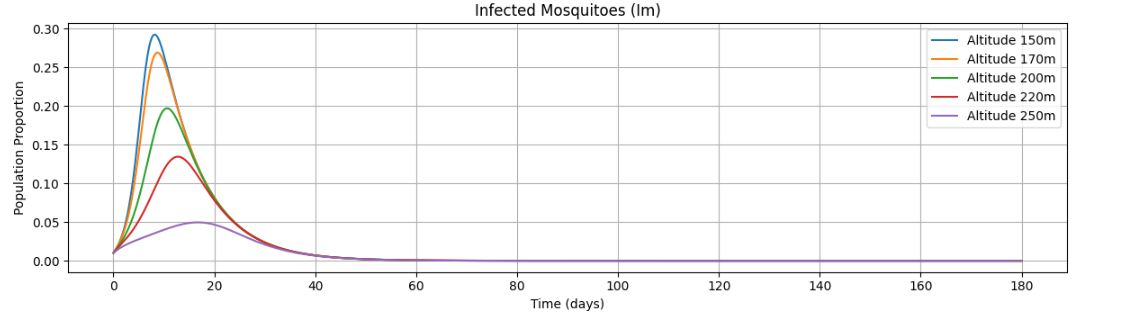}
        \caption{Infected mosquito}
        \label{fig:subfigb36}
    \end{subfigure}
    \caption{Dynamics of mosquito popualtion when temperature is \( 42\) degrees Celsius}
    \label{F}
\end{figure}
 
We also observe that the transmission achieves its highest value when the height is 150 m, the lowest value among the chosen values. To make sure 150 m is the threshold value for the altitude. Here values of altitude considered are 80 m and 100 m. The effect of height on the transmission rate while maintaining a constant temperature of $28^\circ$ can be observed in Figures~\ref{G} and~\ref{H}, while at $35^\circ$ it can be seen in Figures~\ref{I} and~\ref{J}, and at $42^\circ$ in Figures~\ref{K} and~\ref{L}.

\begin{figure}[H]
    \centering
    \begin{subfigure}[b]{0.45\textwidth}
        \centering
        \includegraphics[width=\textwidth]{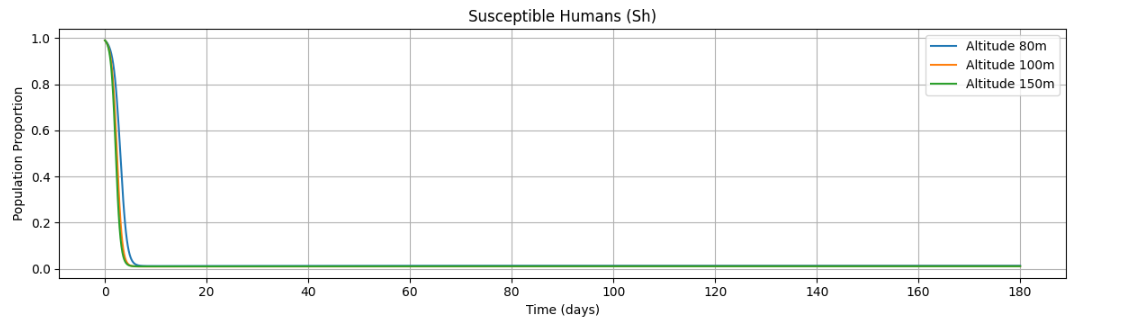}
        \caption{Susceptible humans}
        \label{fig:subfiga37}
    \end{subfigure}
    \hfill
    \begin{subfigure}[b]{0.45\textwidth}
        \centering
        \includegraphics[width=\textwidth]{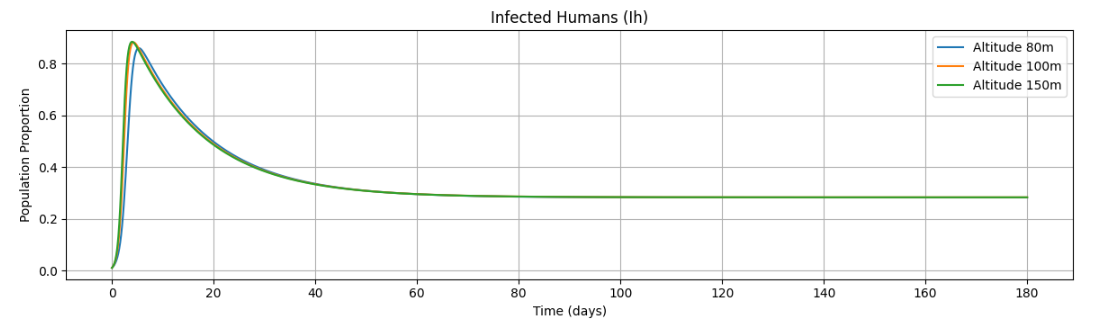}
        \caption{Infected human}
        \label{fig:subfigb38}
    \end{subfigure}
     \caption{Dynamics of human population when  temperature is \( 28 \) degrees Celsius}
    \label{G}
\end{figure}
\begin{figure}[H]
    \centering
    \begin{subfigure}[b]{0.45\textwidth}
        \centering
        \includegraphics[width=\textwidth]{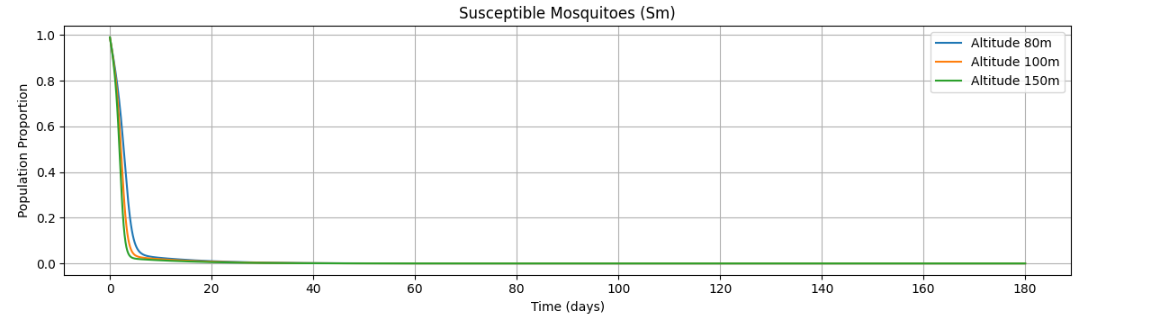}
        \caption{Susceptible mosquito}
        \label{fig:subfiga39}
    \end{subfigure}
    \hfill
    \begin{subfigure}[b]{0.45\textwidth}
        \centering
        \includegraphics[width=\textwidth]{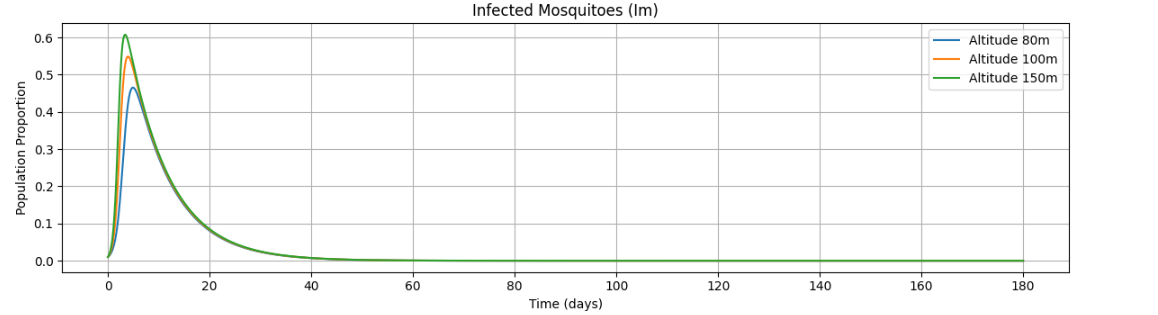}
        \caption{Infected mosquito}
        \label{fig:subfigb40}
    \end{subfigure}
    \caption{Dynamics of mosquito population when temperature is \( 28 \) degrees Celsius}
    \label{H}
\end{figure}
\begin{figure}[H]
    \centering
    \begin{subfigure}[b]{0.45\textwidth}
        \centering
        \includegraphics[width=\textwidth]{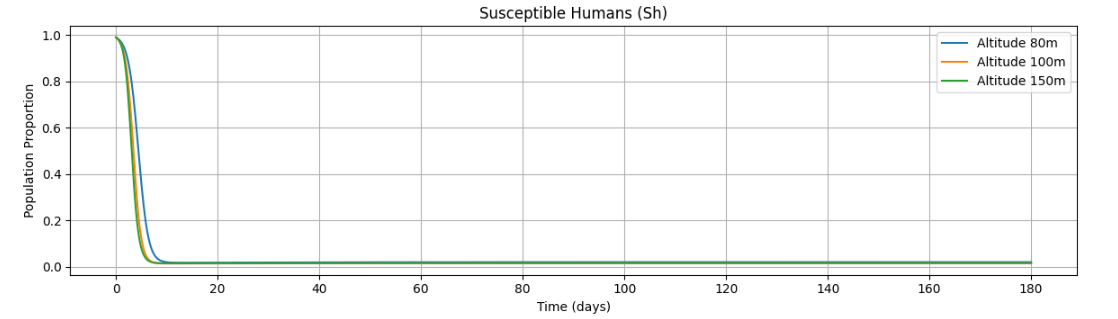}
        \caption{Susceptible human}
        \label{fig:subfiga41}
    \end{subfigure}
    \hfill
    \begin{subfigure}[b]{0.45\textwidth}
        \centering
        \includegraphics[width=\textwidth]{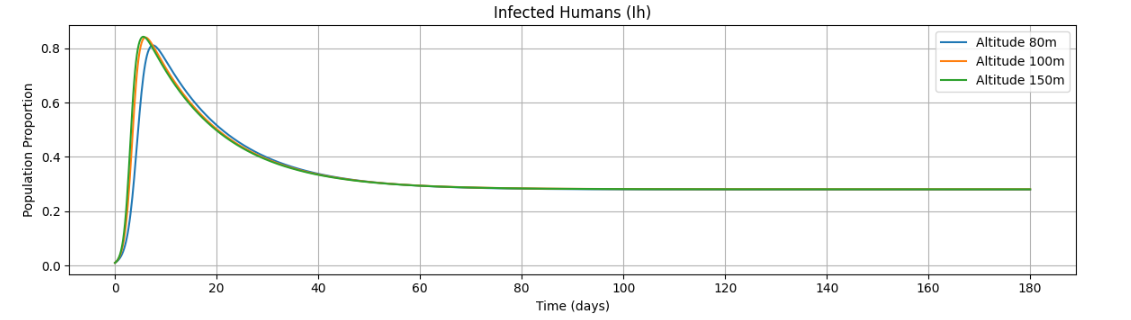}
        \caption{Infected human}
        \label{fig:subfigb42}
    \end{subfigure}
    \caption{Dynamics of human population when temperature is \( 35\) degrees Celsius}
    \label{I}
\end{figure}
\begin{figure}[H]
    \centering
    \begin{subfigure}[b]{0.45\textwidth}
        \centering
        \includegraphics[width=\textwidth]{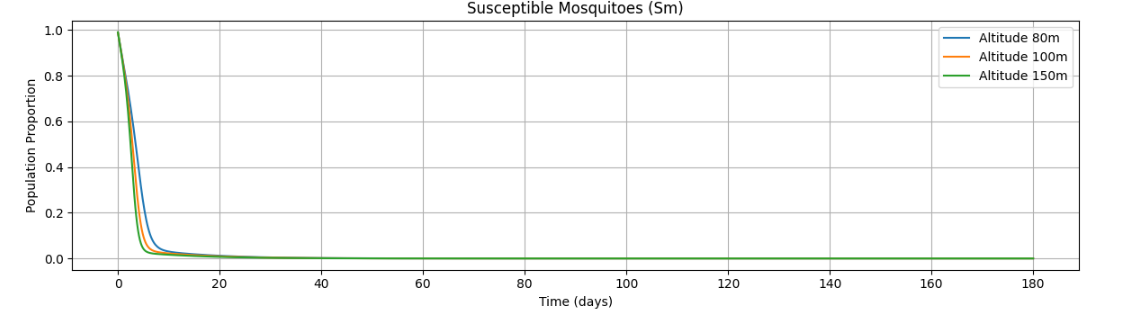}
        \caption{Susceptible mosquito}
        \label{fig:subfiga43}
    \end{subfigure}
    \hfill
    \begin{subfigure}[b]{0.45\textwidth}
        \centering
        \includegraphics[width=\textwidth]{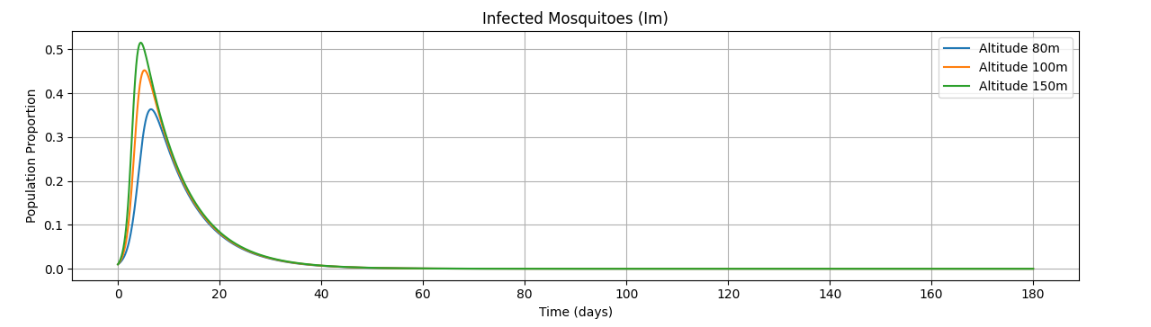}
        \caption{Infected mosquito}
        \label{fig:subfigb44}
    \end{subfigure}
    \caption{Dynamics of mosquito population when temperature is \( 35\) degrees Celsius}
    \label{J}
\end{figure}
 
\begin{figure}[H]
    \centering
    \begin{subfigure}[b]{0.45\textwidth}
        \centering
        \includegraphics[width=\textwidth]{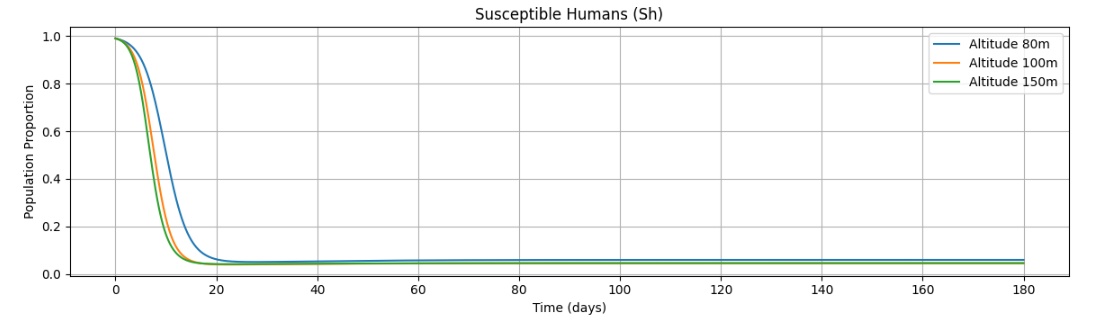}
        \caption{Susceptible human}
        \label{fig:subfiga45}
    \end{subfigure}
    \hfill
    \begin{subfigure}[b]{0.45\textwidth}
        \centering
        \includegraphics[width=\textwidth]{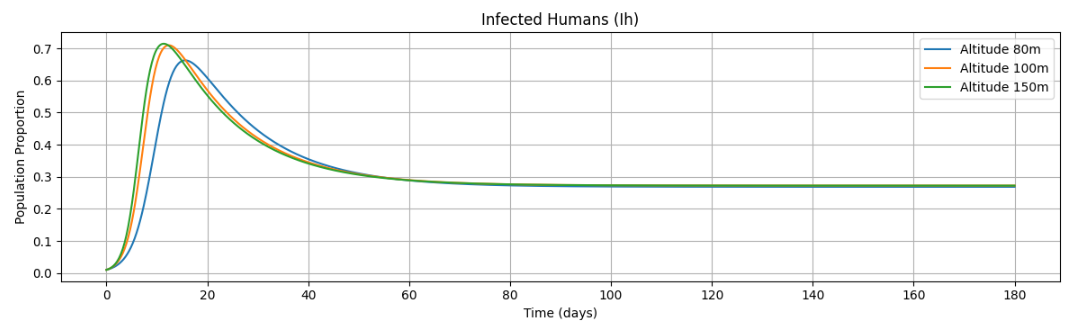}
        \caption{Infected human}
        \label{fig:subfigb46}
    \end{subfigure}
     \caption{Dynamics of human population when temperature is \( 42\) degrees Celsius}
    \label{K}
\end{figure}
\begin{figure}[H]
    \centering
    \begin{subfigure}[b]{0.45\textwidth}
        \centering
        \includegraphics[width=\textwidth]{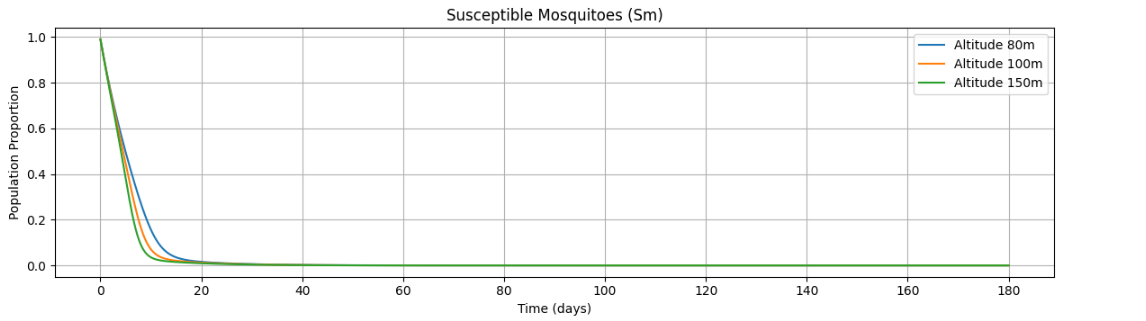}
        \caption{Susceptible mosquito}
        \label{fig:subfiga47}
    \end{subfigure}
    \hfill
    \begin{subfigure}[b]{0.45\textwidth}
        \centering
        \includegraphics[width=\textwidth]{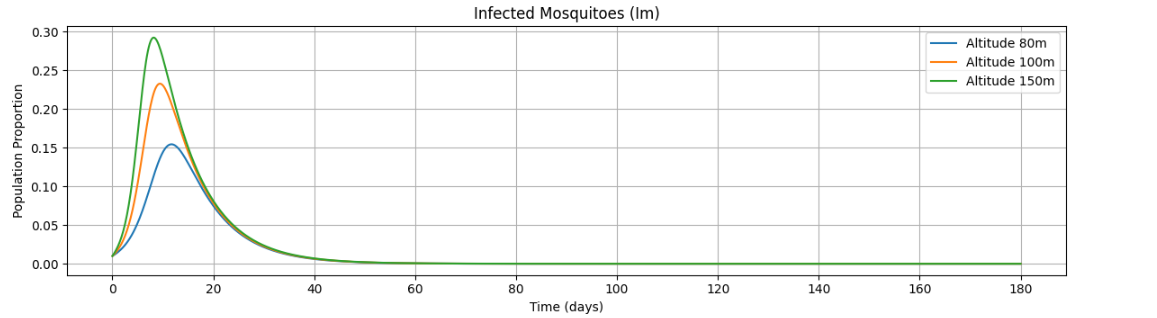}
        \caption{Infected mosquito}
        \label{fig:subfigb48}
    \end{subfigure}
    \caption{ Dynamics of mosquito population when temperature is \( 42\) degrees Celsius}
    \label{L}
\end{figure}

Thus we observe that the transmission rate is the highest when the height is $150$ and in this case, it is the highest value among the chosen values thus from this we can conclude that $h = 150$  is the threshold value for the altitude.  
The function 
$f(h) = e^{-\frac{h^{2}}{\xi^{2}}}(1-e^{-\frac{h^{2}}{\xi^{2}}})$
achieves its highest value when 
$h = \xi{\sqrt{\ln(2)}}$
and thus the transmission rate increases as the value of height approaches $h = {\xi\sqrt{\ln(2)}}$. By substituting both the values of $\xi$ and taking the average we get approximately $142.1$ thus the optimal value for altitude can be concluded to be in the range $140 - 150$ and since it can be observed that the trajectory for all of these values is nearly the same. \\
\newline 
\section{Parameter estimation}
An important factor influencing disease transmission is the parameters associated with the transmission dynamics. Understanding these parameters is crucial for medical professionals, as it enables them to assess the severity of the disease and implement appropriate interventions. This section presents three distinct neural network models: Artificial Neural Networks (ANNs), Recurrent Neural Networks (RNNs), and Physics-Informed Neural Networks (PINNs) for estimating these parameters. Utilizing the predicted parameters, we can forecast the trajectories of the various compartments involved in the disease dynamics.
 
 For this study, the Artificial Neural Networks (ANNs) utilized comprised a total of five layers, including three hidden layers. All layers, except for the output layer, contained 15 dense units and employed the sigmoid function as the activation function. The output layer featured seven dense units with no activation function applied. The Recurrent Neural Networks (RNNs) architecture implemented in this work consisted of three layers: one input layer, one dropout layer, and one output layer. The input layer was composed of 50 LSTM units with the ReLU activation function. The dropout layer had a dropout rate of 20 \%, while the output layer contained seven dense units without an activation function.
The Physics-Informed Neural Networks (PINNs) architecture mirrored that of the ANN, with the primary distinction being the number of nodes in the input and output layers, which were one and five, respectively, in the PINNs architecture.

\subsection{Methodology}
Since three different architectures of neural networks are used there is a difference between methodologies used for ANN, RNN, and PINN. For ANN and RNN, both of these models have been trained on a training dataset that has 1000 data points where the input is the first 10 points of the trajectories of all of the compartments and the output is the corresponding parameters. For PINN the methodology is completely different from the earlier models since here there is no involvement of a training dataset but instead an assumed set of parameters is chosen and by minimizing the loss function the assumed set of parameters is made to converge to the actual set of parameters by comparing the obtained and actual trajectories. Here the input is the time whereas the output is the trajectories of all the five compartments.

For a novel comparison between the models, each model is simulated for 20000 epochs. Details of actual and predicted parameters are given in the tables below.  
\begin{table}[H]
\centering
\begin{tabular}{|c|c|c|c|}
\hline
\textbf{Parameter} & \textbf{Predicted Value} & \textbf{Actual Value} & \textbf{Relative Error (\%)} \\ \hline
Birth rate of humans & 0.38273 & 0.40000 & 4.32 \\ \hline
Transmission rate of humans & 0.28249 & 0.30000 & 5.84 \\ \hline
Death rate of humans & 0.09973 & 0.10000 & 0.27 \\ \hline
Recovery rate of humans & 0.04640 & 0.01000 & 364.01 \\ \hline
Birth rate of mosquitoes & 0.07613 & 0.05000 & 52.25 \\ \hline
Transmission rate of mosquitoes & 0.01045 & 0.02000 & 47.73 \\ \hline
Death rate of mosquitoes & 0.05757 & 0.04000 & 43.92 \\ \hline
\end{tabular}
\caption{Comparison of predicted and actual values with relative errors for ANN}
\label{tab:relative_errors}
\end{table}

\begin{table}[H]
\centering
\begin{tabular}{|c|c|c|c|}
\hline
\textbf{Parameter} & \textbf{Predicted Value} & \textbf{Actual Value} & \textbf{Relative Error (\%)} \\ \hline
Birth rate of humans & 0.19740 & 0.40000 & 50.65 \\ \hline
Transmission rate of humans & 0.18619 & 0.30000 & 37.94 \\ \hline
Death rate of humans & 0.34348 & 0.10000 & 243.48 \\ \hline
Recovery rate of humans & 0.21954 & 0.01000 & 2095.38 \\ \hline
Birth rate of mosquitoes & 0.18898 & 0.05000 & 277.96 \\ \hline
Transmission rate of mosquitoes & 0.24995 & 0.02000 & 1149.77 \\ \hline
Death rate of mosquitoes & 0.38605 & 0.04000 & 865.13 \\ \hline
\end{tabular}
\caption{Comparison of predicted and actual values with relative errors for RNN}
\label{tab:updated_relative_errors}
\end{table}

\begin{table}[H]
\centering
\begin{tabular}{|c|c|c|c|}
\hline
\textbf{Parameter} & \textbf{Predicted Value} & \textbf{Actual Value} & \textbf{Relative Error (\%)} \\ \hline
Birth rate of humans & 0.39200 & 0.40000 & 2.00 \\ \hline
Transmission rate of humans & 0.28186 & 0.30000 & 6.05 \\ \hline
Death rate of humans & 0.09793 & 0.10000 & 2.07 \\ \hline
Recovery rate of humans & 0.00949 & 0.01000 & 5.07 \\ \hline
Birth rate of mosquitoes & 0.04593 & 0.05000 & 8.14 \\ \hline
Transmission rate of mosquitoes & 0.01969 & 0.02000 & 1.56 \\ \hline
Death rate of mosquitoes & 0.03600 & 0.04000 & 10.01 \\ \hline
\end{tabular}
\caption{Comparison of predicted and actual values with relative errors for PINN}
\label{pinn}
\end{table}

\subsection{Prediction of trajectories}
We observe that out of all the models PINNs can be observed to be superior to the remaining models. One of the main reasons why PINNs proved to be a superior model is because disease transmissions are not random phenomena but are phenomena following a certain law and thus PINNs not only model the data but also model the associated physics law thus enabling the models to make predictions with very high accuracy. The parameters predicted by PINNs and the trajectories of all five compartments are predicted. The evolution of the human population can be observed in Figure \ref{cell1} and the evolution of the mosquito population can be observed in Figure \ref{cell2}.

%%%%%%%%%%%%%%%%%%%%%%%%%%%%%%%%%
\begin{figure}[H] 
	\centering
	\begin{subfigure}{.45\textwidth} 
	\centering
	\includegraphics[scale=0.32]{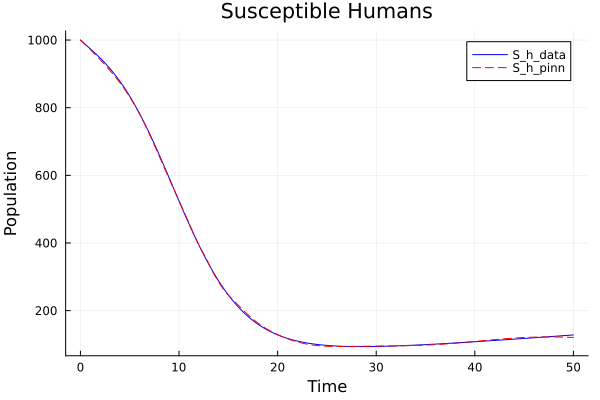}
		\caption{}
		\label{cell11}
	\end{subfigure}
%%%%%%%%%%%%%%
	\begin{subfigure}{.45\textwidth}
	\centering
		\includegraphics[scale=0.32]{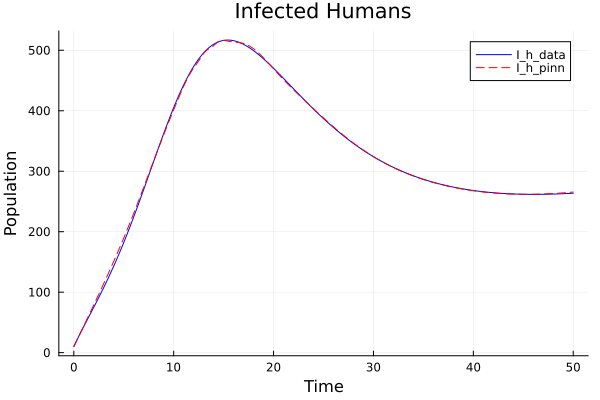}
		\caption{}
		\label{cell12}
	\end{subfigure}
	\caption{Prediction made by PINNs for the human population.}
	\label{cell1}
\end{figure} 
%%%%%%%%%%%%%%%%%%%%%%%%%%%%%%%%%

%%%%%%%%%%%%%%%%%%%%%%%%%%%%%%%%%
\begin{figure}[H] 
	\centering
	\begin{subfigure}{.45\textwidth} 
	\centering
	\includegraphics[scale=0.32]{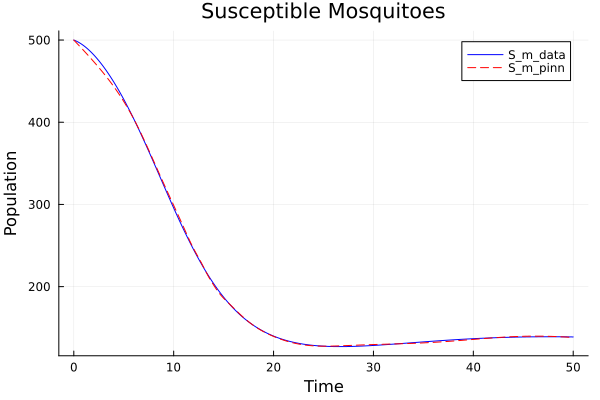}
		\caption{}
		\label{cell11}
	\end{subfigure}
%%%%%%%%%%%%%%
	\begin{subfigure}{.45\textwidth}
	\centering
		\includegraphics[scale=0.32]{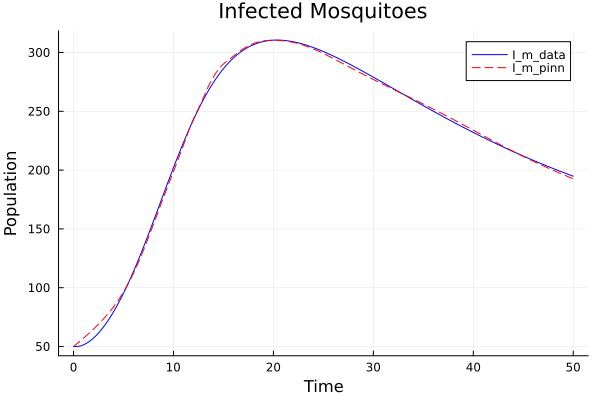}
		\caption{}
		\label{cell12}
	\end{subfigure}
	\caption{Prediction made by PINNs for the mosquito population.}
	\label{cell2}
\end{figure} 
%%%%%%%%%%%%%%%%%%%%%%%%%%%%%%%%%

\subsection{Finding the risk of a disease}  
The most important aspect of a disease is the determination of risk. Whenever there is a disease outbreak in a country there are some regions where there is more risk compared to the other regions thus it is essential to calculate this risk of every region. This problem statement is addressed using the method of DMD (dynamic mode decomposition) and the main reason for using this method is because  DMD can make exact predictions from raw data, unlike other deep learning methods. The complete methodology of the problem statement can be seen in  Figure  ~\ref{DMD}.\\
 
In this work, DMD is used to calculate the disease risk in a particular region. DMD gives the oscillations of the dynamics and thus using the peak values of DMD will give us the overall oscillations of the dynamics which is nothing but the measure of risk. The DMD plot and the eigenvalue spectrum can be found in Figure  ~\ref{DMD1}.\\
 
From the eigenvalue spectrum, we can observe that all the points are either on or within the unit circle. This shows that the transmission of malaria in Africa is stable. African map with the corresponding color coding based on the risk can be found in Figure  ~\ref{AM}.\\
\begin{figure}[H]
    \centering
    \includegraphics[scale = 0.55]{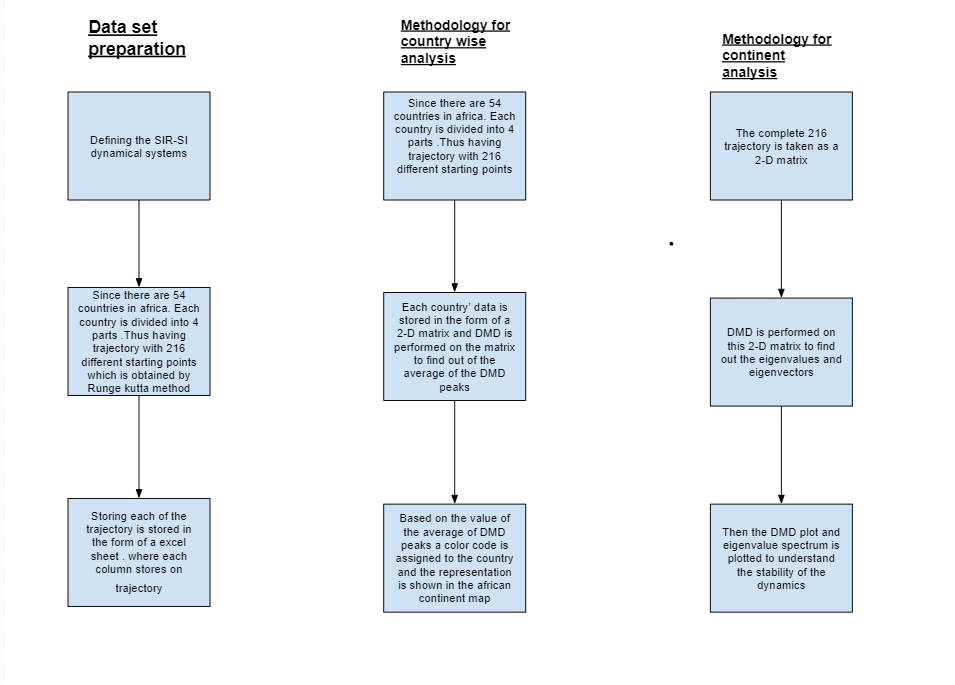}
    \caption{Flow chart of methodology}
    \label{DMD}
\end{figure}
\begin{figure}[H]
    \centering
    \begin{subfigure}[b]{0.45\textwidth}
        \centering
        \includegraphics[width=\textwidth]{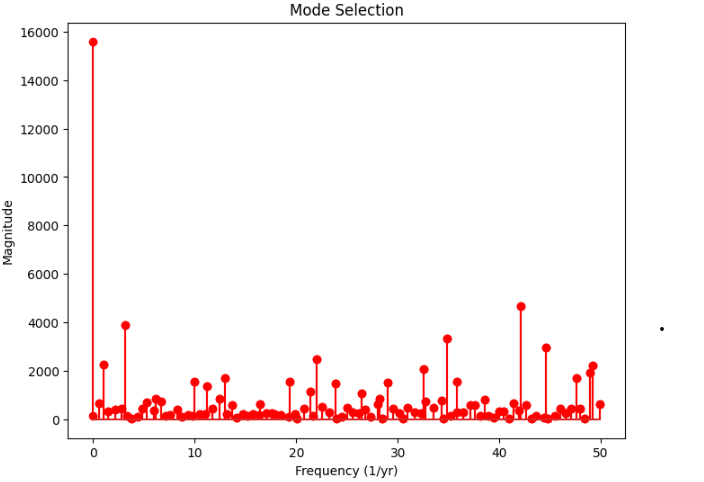}
        \caption{DMD plot}
        \label{fig:subfiga51}
    \end{subfigure}
    \hfill
    \begin{subfigure}[b]{0.45\textwidth}
        \centering
        \includegraphics[width=\textwidth]{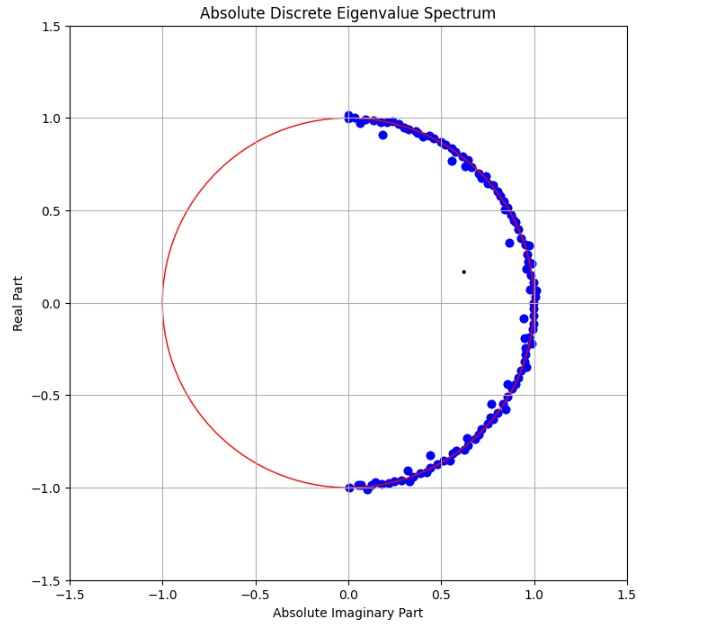}
        \caption{Infected human}
        \label{fig:subfigb52}
    \end{subfigure}
    \caption{DMD plot and the eigenvalue spectrum}
    \label{DMD1}
\end{figure}

\begin{figure}[H]
    \centering
    \includegraphics[scale=0.5]{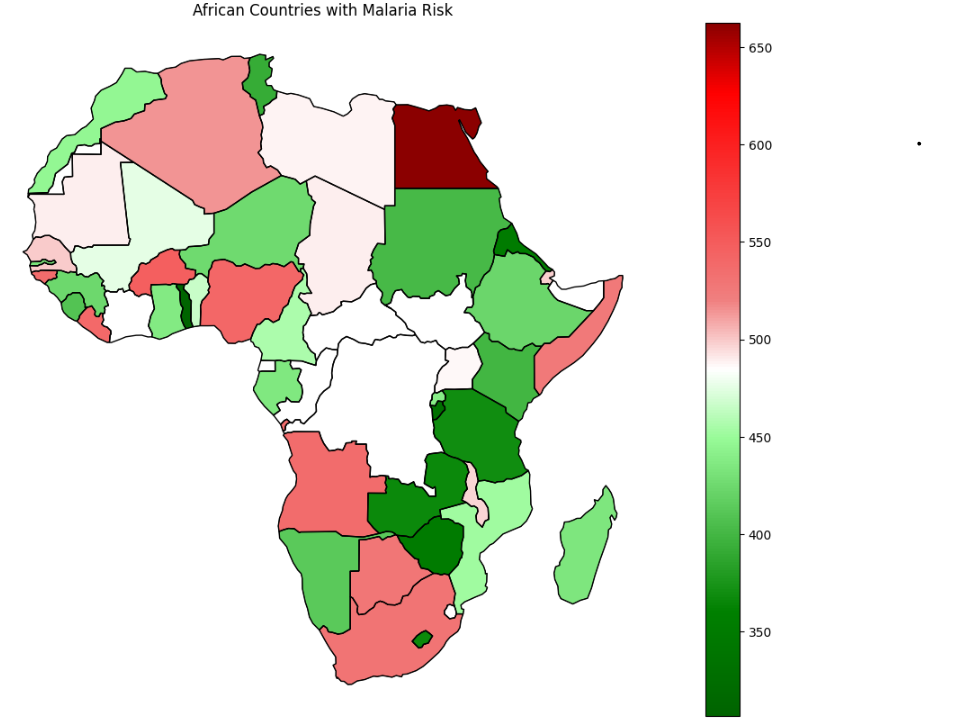}
    \caption{African Map with color coding based on risk}
    \label{AM}
\end{figure}
\section{Concluding remarks}
In this work, an attempt is made to understand the dynamics of malaria transmission using various mathematical and machine-learning techniques. The proposed theorems related to steady states were validated using numerical simulations it was followed by Parameter estimation using three different architectures of neural networks namely ANNs, RNNs, and PINNs, and the predicted parameter of the best model was then utilized to predict the trajectories of all the compartments. Finally, this work was concluded by finding the measure of risk using DMD.\\
 
This mathematical model developed in this work can be deployed by the government of any country to analyze the trajectory of infected people whenever there is a disease outbreak and then predict the future state of the country which will give medical professionals an idea of what kind of measures must be taken to reduce the disease spread \\
 
In the future, our idea is to understand the dynamics of malaria using techniques of physics-informed machine learning such as  Sparse identification of non-linear dynamics (SINDy), Universal differential equations(UDE), Neural differential equations (NeuralODE), and other variants of PINN's like VPINNs and Recurrent PINNs. We intend to predict the trajectories of the compartments using time series models such as Attention models, ResNet architecture, etc using a publicly available dataset for epidemic forecasting. We also want to study the stochastic version of this model mainly to study the random fluctuations in the birth and mortality process. 

\section{Data availability statement}
This study uses solely synthetic data, which was generated for the purpose of this research. The synthetic data is not based on real-world observations and can be freely accessed. The dataset is available at the \href{https://drive.google.com/drive/folders/1y4BdhmRC8sl81OwzrdmYZgVWns946R9X?usp=sharing}{link}.
\bibliographystyle{unsrturl}

\end{document}